\documentclass[11pt, oneside]{article}   	
\usepackage{geometry}                	
\usepackage{graphicx}	
\usepackage{amssymb}
\usepackage{float}
\usepackage[cmex10]{amsmath}
\usepackage{amsthm}
\usepackage{soul,color}

\newcommand\bbR{\mathbb R}
\newcommand\bx{\boldsymbol x}

\newcommand\xb{\boldsymbol x}
\newcommand\yb{\boldsymbol y}
\newcommand\bE{{\bf E}}
\newcommand\E{{\bf E}}

\newcommand\bH{{\bf H}}
\newcommand\bA{\boldsymbol A}
\newcommand\A{\boldsymbol A}

\newcommand\J{\boldsymbol J}

\newcommand\K{\boldsymbol K}

\newcommand\bb{\boldsymbol b}
\newcommand\ba{\boldsymbol a}

\newcommand\In{\operatorname{inc}}

\newcommand\bn{\boldsymbol n}

\newtheorem{remark}{Remark}
\newtheorem{definition}{Definition}
\newtheorem{thm}{Theorem}
\newtheorem{lemma}{Lemma}

\title{Decoupled field integral equations for 
electromagnetic scattering from homogeneous penetrable obstacles}
\author{Felipe Vico\footnote{
Instituto de Telecomunicaciones y Aplicaciones Multimedia (ITEAM),
Universidad Polit\` ecnica
de Val\` encia, 46022 Val\` encia, Spain. {{\em email}:\ {\tt
{felipe.vico@gmail.com}}}, {{\sf {mferrand@dcom.upv.es.}}} },
Leslie Greengard\footnote{Courant Institute, New York University, New
  York, NY and Center for Computational Biology, Flatiron Institute, 
New York, NY. {\em email}:
  {\tt greengard@cims.nyu.edu}}, and
Miguel Ferrando\footnotemark[1]
}

\begin{document}
\maketitle
\begin{abstract}
We present a new method for the analysis of 
electromagnetic scattering from homogeneous penetrable bodies. 
Our approach is based on a reformulation of the governing 
Maxwell equations in terms of two uncoupled vector Helmholtz systems: 
one for the electric field and one for the magnetic field.
This permits the derivation of 
resonance-free Fredholm equations of the second kind that
are stable at all frequencies,
insensitive to the genus of the scatterers, and invertible for
all passive materials including those with negative permittivities 
or permeabilities.
We refer to these as {\em decoupled field integral equations}.
\end{abstract}

\section{Introduction}

A standard problem in electromagnetics concerns the solution of 
of the Maxwell equations in an exterior domain, when an 
incoming wave is scattered by a collection
of bounded penetrable obstacles.
There is a vast literature on integral equation methods for such problems,
with standard formulations for the piecewise constant case 
described, for example, in the papers \cite{Glisson,Marx,Mautz,Muller,PMCHW}.
Boundary integral methods are natural in the homogeneous
(piecewise constant) setting, since they
require only the discretization of the inclusion boundaries rather than the
domain, satisfy the Silver-M\"{u}ller radiation condition exactly,
and can be coupled with suitable fast algorithms,
such as the fast multipole method (FMM). 

Working in the frequency domain and assuming a time dependence 
of the form $e^{-i\omega t}$, Maxwell's equations 
in a linear, isotropic material 
are given by 
\begin{align}
  \nabla \times {\bf H}^{tot} &= -i\omega \epsilon_L \,  {\bf E}^{tot},
\label{meq1} \\
  \nabla \times {\bf E}^{tot} &= i\omega \mu_L {\bf H}^{tot},  \nonumber
\end{align}
where
$\epsilon_L$ and $\mu_L$ are the local permittivity and permeability of the medium,
respectively. 
Here, ${\bf E}^{tot}$ and ${\bf H}^{tot}$ denote the total 
electric and magnetic fields which we express in the form 
\begin{equation}
  {\bf E}^{tot} = {\bf E}^{in} + {\bf E}_0, \quad
  {\bf H}^{tot} = {\bf H}^{in} + {\bf H}_0  \label{fbgvi3a} ,
\end{equation}
in the exterior region and 
\begin{equation}
  {\bf E}^{tot} = {\bf E}, \quad
  {\bf H}^{tot} = {\bf H}  \label{fbgvi3ain} ,
\end{equation}
within the inclusions. Here, $\{ {\bf E}^{in}, {\bf H}^{in} \}$ is a
known {\em incoming} electromagnetic field.

For the sake of simplicity, we will assume that the penetrable obstacles
are all made of the same material and that they
define a compactly supported open region $D$ in $\mathbb{R}^3$ 
whose boundary, denoted by  $\partial D$, consists of 
a finite number of disjoint, closed 
surfaces belonging to class $C^2$. We also assume that the 
exterior region $\mathbb{R}^3/\overline{D}$ is connected.
We will denote by $\epsilon_0,\mu_0$ the material properties of the exterior
domain and by $\epsilon,\mu$ the material properties of the inclusions.

At obstacle boundaries, the Maxwell equations (\ref{meq1}) must satisfy 
the continuity conditions
\cite{Jackson,Papas}:
\begin{equation}
\label{mwaveint1} 
\begin{aligned}
 {\bf n} \times ({\bf E}_0 - \bE ) &= -
  {\bf n} \times {\bf E}^{in} 
    \\
 {\bf n} \times ({\bf H}_0 - \bH ) &= -
  {\bf n} \times {\bf H}^{in},
\end{aligned}
\end{equation}
where ${\bf n}$ denotes the outward normal vector to the interface.
\begin{definition}
The boundary value problem defined by \eqref{meq1} with 
$\{ \epsilon_L,\mu_L \} = \{ \epsilon,\mu \}$ inside the inclusions and 
$\{ \epsilon_L,\mu_L \} = \{ \epsilon_0,\mu_0 \}$ in the exterior domain,
together with the interface conditions \eqref{mwaveint1},
will be referred to as the 
{\em Maxwell transmission problem}.
\end{definition}

It is well-known that, under mild assumptions on the 
permittivity and permeability
\cite{homogeneous_uniqueness_3,Muller},
the Maxwell transmission problem has a unique solution, so long as 
the scattered field satisfied the
Silver-M{\" u}ller radiation condition:
\begin{equation}\label{radEH}
\begin{array}{ll}
\sqrt{\frac{\mu_0}{\epsilon_0}}\bH_0(\bx)\times\frac{\bx}{|\bx|}-\bE_0(\bx)=o\Big(\frac{1}{|\bx|}\Big), \quad &|\bx|\rightarrow \infty\, .
\end{array}
\end{equation}
We will restrict our attention in this paper to the case
$\omega \geq 0$ with $\epsilon_0,\mu_0>0$ and
\begin{equation}\label{microwave_cond}
\begin{aligned}
\Im{(\epsilon)}&>0 \text{ or } \epsilon \in \bbR^+ \\
\Im{(\mu)}&> 0 \text{ or } \mu \in \bbR^+ ,
\end{aligned}
\end{equation}
where $\bbR^+$ denotes the positive real numbers.
For the sake of completeness, we provide a simple proof of uniqueness
for this parameter regime (Appendix \ref{maxuniq}), which includes
all metamaterials wth non-zero dissipation.

The standard approach to the development of integral equation methods is 
based on the classical vector and scalar potentials
\cite{Papas}:
We let
\[ \A_0(\xb) = \mu_0 S_{k_0}[\J_0](\xb),\ \ 
\A(\xb) = \mu S_{k}[\J](\xb), \]
where
\begin{equation}
S_{k}[\J](\xb) \equiv  \int_{\partial D}  g_k(\xb-\yb) \, \J(\yb) \, 
ds_{\yb} \, , 
\label{Skvecdef}
\end{equation}
$g_k(\xb) = e^{i k |\xb\|}/\|\xb\|$,
$k = \sqrt{ w^2 \epsilon \mu}$, 
and $k_0 = \sqrt{ w^2 \epsilon_0 \mu_0}$, 
Here, $\J,\J_0$ can be viewed as surface electric currents.
When the argument of the square root is 
complex, $k$ is taken to lie in the upper half-plane.
We define the vector anti-potentials by 
\[ {\tilde \A}_0(\xb) = \epsilon_0  S_{k_0}[\K_0](\xb),\ \ 
 {\tilde \A}(\xb) = \epsilon S_{k}[\K](\xb),
\]
where $\K,\K_0$ can be viewed as surface magnetic currents.
The scalar potentials and antipotentials are given by 
\[ \phi_0(\xb) = 
\frac{1}{i\omega \epsilon_0\mu_0} \nabla \cdot \A_0,
\ \ \phi(\xb) =  
\frac{1}{i\omega \epsilon\mu} \nabla \cdot \A,
\]
and 
\[ \psi_0(\xb) = 
\frac{1}{i\omega \epsilon_0\mu_0} \nabla \cdot {\tilde  \A}_0,
\ \ \psi(\xb) = 
\frac{1}{i\omega \epsilon\mu} \nabla \cdot {\tilde  \A}.
\]
From these, we may write
\begin{equation}
\label{Eeval}
\begin{aligned}
\E_0 &= i \omega \A_0 - \nabla \phi_0 - 
\frac{1}{\epsilon_0} \nabla \times {\tilde \A}_0,  \\
\bH_0 &= \frac{1}{\mu_0} \nabla \times \A_0 +
i \omega {\tilde \A}_0 - \nabla \psi_0
\end{aligned}
\end{equation}
for $\xb$ in the exterior, and 
\begin{equation}
\label{Heval}
\begin{aligned}
\E &= i \omega \A  -\nabla \phi - 
       \frac{1}{\epsilon} \nabla \times {\tilde \A},  \\
\bH &= \frac{1}{\mu} \nabla \times \A 
+ i \omega {\tilde \A} - \nabla \psi
\end{aligned}
\end{equation}
for $\xb$ inside the inclusions.

The number of degrees of freedom in the representation
is reduced by assuming that 
the tangential vector fields
$\J_0,\J, \K_0, \K$ satisfy
\begin{equation}
\J^{s}:= \epsilon_0 \J_0 =  \epsilon \J   \qquad 
\K^{s} := \mu_0 \K_0= \mu \K   \, .
\label{J0def}
\end{equation}
Imposing the conditions 
(\ref{mwaveint1}) and using the relations \eqref{rhodef}
yields M\"{u}ller's integral equation \cite{Muller} for $\J^s,\K^s$:

{\small
\begin{eqnarray*}
-\bn \times \E^{\In} = - \left( \frac{\mu_0 + \mu}{2} \right)\K^s-\left(\mu_0\nabla\times S_{k_0}-\mu\nabla\times S_{k}\right)[\K^s]+i\omega\left( \epsilon_0\mu_0 S_{k_0}- \epsilon\mu S_{k}\right)[\J^s]\\
-\frac{1}{i\omega}\left(\nabla\nabla\cdot S_{k_0}-\nabla\nabla\cdot S_{k}\right)[\J^s]
\nonumber
\end{eqnarray*}

\begin{eqnarray*}
-\bn \times \bH^{\In} =  \left( \frac{\epsilon_0 + \epsilon}{2} \right)\J^s+\left(\epsilon_0\nabla\times S_{k_0}-\epsilon\nabla\times S_{k}\right)[\J^s]+i\omega\left( \epsilon_0\mu_0 S_{k_0}- \epsilon\mu S_{k}\right)[\K^s]\\
-\frac{1}{i\omega}\left(\nabla\nabla\cdot S_{k_0}-\nabla\nabla\cdot S_{k}\right)[\K^s]
\nonumber
\end{eqnarray*}
}

In M\"{u}ller's original formulation, the unknowns are 
actually the tangential components of $\bE$ and $\bH$ rather than fictitious 
currents, and the integral equation above is the 
dual of the formulation in \cite{Muller}. 
We refer the reader to
\cite{CK1,CK2,Calderon_PMCHW,homogeneous_uniqueness_3,FMPS,
homogeneous_uniqueness_8,Yla-Oijala}
for further details. 

It suffices, for our present purposes, to 
note that all of the operators appearing above turn out to be compact and
that M\"{u}ller's integral equation is a resonance-free 
Fredholm equation of the second kind. It is 
uniquely solvable for the passive materials under consideration here.
Moreover, it can be shown that 
the integral equation is stable even in the low-frequency 
regime.  Unfortunately, the representation 
itself is subject to ``low-frequency breakdown". That is, once the 
currents are known, the evalution of $\bE$ or $\bH$ 
from \eqref{Eeval}, \eqref{Heval} is unstable as
$\omega \rightarrow 0$, because of catastrophic cancellation in the 
scalar potentials.

Rather than writing the scalar 
potentials and antipotentials as above, however, we may write
\begin{equation}
\phi_0(\xb) = 
  S_{k_0}[\rho](\xb) 
\ \ \phi(\xb) =  
  S_k[\rho](\xb), 
\label{slpdef1}
\end{equation}
and 
\begin{equation}
\psi_0(\xb) = 
 S_{k_0}[\rho_M](\xb),
\ \ \psi(\xb) = 
 S_k[\rho_M](\xb),
\label{slpdef2}
\end{equation}
where
\begin{equation}
S_k[\rho](\xb) \equiv \int_{\partial D}  g_k(\xb-\yb) \, 
\rho(\yb) \, ds_{\yb} \, ,
\label{Skscalardef}
\end{equation}
and 
\begin{equation}
\label{rhodef}
\rho = \nabla_S \cdot \J/(i \omega), \quad
\rho_M = \nabla_S \cdot \K/(i \omega).
\end{equation}
Here, $\nabla_S$ denotes the surface divergence operator.

\begin{remark}
The problem of low-frequency breakdown lies in the computation of 
$\rho, \rho_M$. Since ill-conditioning occurs even for a 
fixed $\omega$ under mesh
refinement, a more descriptive term is, perhaps, dense-mesh breakdown
\cite{Valdes}, which we will use interchangeably.
\end{remark}

One mechanism to overcome dense-mesh breakdown is to solve two 
auxiliary integral equations for the scalar potentials 
\cite{MullerAuxDavos},
using the representation
\eqref{Eeval} and \eqref{Heval} and imposing the 
two interface conditions
\[ \bn \cdot (\epsilon_0 \E^{tot}_0 - \epsilon \E^{tot}) = 0, \]
\[ \frac{\nabla \cdot \E_0}{k_0^2} - 
\frac{\nabla \cdot \E}{k^2} = 0. \]
The first is a standard condition for dielectric interfaces and leads
to the scalar equation
\begin{equation}
\epsilon_0 \frac{\partial \phi_0}{\partial n}  - 
\epsilon \frac{\partial \phi}{\partial n} =  f
\label{aug1}
\end{equation}
where
\[ f = \bn \cdot \E^{in} + 
i \omega \bn \cdot (\epsilon_0^2\mu_0 \cdot S_{k_0}[\J_0] 
- \epsilon^2\mu \cdot S_{k}[\J]) 
+ \bn \cdot \nabla \times (\epsilon\mu S_{k}[\K] -
 \epsilon_0\mu_0 S_{k_0}[\K_0]) 
\]
The second interface condition
follows from the fact that $\E$ and $\E_0$ are divergence-free,
with the particular linear combination chosen to yield the equation
\begin{equation}
\phi_0 - \phi =  \frac{\nabla \cdot S_{k}[\J] -
 \nabla \cdot S_{k_0}[\J_0]}{i \omega}.
\label{aug2}
\end{equation}
Similarly, one can derive the two scalar equations for 
$\psi$ and $\psi_0$:
\begin{equation}
\mu_0 \frac{\partial \psi_0}{\partial n}  - 
\mu \frac{\partial \psi}{\partial n} =  g
\label{aug1H}
\end{equation}
where
\[ g = \bn \cdot \bH^{in} + 
i \omega \bn \cdot (\mu_0^2\epsilon_0 \cdot S_{k_0}[\K_0] 
- \mu^2\epsilon \cdot S_{k}[\K]) 
+ \bn \cdot \nabla \times (\epsilon\mu S_{k}[\J] -
 \epsilon_0\mu_0 S_{k_0}[\J_0]),
\]
and
\begin{equation}
\psi_0 - \psi =  \frac{\nabla \cdot S_{k}[\K] -
 \nabla \cdot S_{k_0}[\K_0]}{i \omega}.
\label{aug2H}
\end{equation}

Without entering into details, it is well-known that the interface
problems \eqref{aug1},\eqref{aug2} and
\eqref{aug1H},\eqref{aug2H} are uniquely solvable without dense-mesh
breakdown using a Fredholm equation of the second kind 
\cite{homogeneous_uniqueness_1,rokhlinTRANS}. (Care must be taken in 
computing the right-hand sides for \eqref{aug2} and \eqref{aug2H}.
In particular, catastrophic cancellation can be avoided for small $\omega$ 
by computing the difference asymptotically.)

\begin{remark}
A solver for the 
Maxwell transmission problem based on first solving the 
M\"{u}ller integral equation, followed by solving 
the systems \eqref{aug1},\eqref{aug2} and
\eqref{aug1H},\eqref{aug2H} will be referred to as a 
{\em decoupled charge-current formulation}.
It yields stable solutions for real $\epsilon,\mu$
\cite{MullerAuxDavos}, but can have resonances for lossy materials.
\end{remark}

It is also worth noting that alternative ``charge-current" formulations
have been developed \cite{CCbad,CCIE3,CCIE1,CCIE2}
that make use of 
electric and magnetic charge as additional, distinct unknowns.
Using the representation 
\eqref{Eeval}, \eqref{Heval} with 
$\A,\A_0,{\tilde \bA},{\tilde \bA}_0$ defined in terms of 
the unknowns $\J,\K$ and the relation \eqref{J0def}, 
with $\phi,\phi_0,\psi,\psi_0$ defined in terms of 
$\rho,\rho_M$, one can seek to enforce
\begin{equation}
\begin{aligned}
 {\bf n} \cdot (\epsilon_0 {\bf E}_0 - \epsilon \E ) &= -
  {\bf n} \cdot \, \epsilon_0 {\bf E}^{in} 
    \\
 {\bf n} \cdot (\mu_0 {\bf H}_0 - \mu \bH ) &= -
  {\bf n} \cdot \mu_0 {\bf H}^{in}, \\
 {\bf n} \times ({\bf E}_0 - \bE ) &= -
  {\bf n} \times {\bf E}^{in} 
    \\
 {\bf n} \times ({\bf H}_0 - \bH ) &= -
  {\bf n} \times {\bf H}^{in}.
\end{aligned}
\end{equation}

This avoids low-frequency breakdown and leads to a Fredholm equation of the 
second kind. However, these methods are subject to spurious 
``near resonances."
The precise location of these resonances depends on the 
specific charge-current scheme employed, the material properties and the 
geometry, illustrated in Fig. \ref{Conditionfig} below.
A ``decoupled potential formulation," presented in  
\cite{DPIE_chew}, extends the method of \cite{DPIE} for perfect conductors
to the dielectric case. It, too, is subject to spurious resonances
when the real parts of $\epsilon$ and 
$\mu$ are not both positive, because of the existence of 
``transmission" eigenvalues \cite{homogeneous_uniqueness_10}.

In \cite{gDEBYE2}, an 
integral representation was developed using four scalar densities
supported on the obstacle boundaries. This ``generalized Debye source" 
representation yields a resonance-free
Fredholm integral equations of the second kind, valid for all the material
properties of interest in the present paper. 
The topology of the domain, however, plays a critical
role and the method requires a basis for surface harmonic vector fields.
Finally, there is a substantial literature on ``single source" integral
equations, which involve only one unknown tangential vector field
(see, for example \cite{Glisson,Marx,singlesource_yeung}). 
Unfortunately, the original single source formulations typically 
do not lead to Fredholm equations,  
involve hypersingular operators, and/or are subject to low-frequency 
breakdown. Significant progress, however, has been made
in formulating equations that are immune from low-frequency
breakdown (see, for example, \cite{Colliander,Mautz,Valdes}).
Unfortunately, none of these equations have been shown to be resonance-free
for the full range of passive materials which we consider here.

In the present paper, we describe a new system of decoupled Fredholm 
integral equations
for the electric and magnetic fields that are resonance-free 
for all problems of interest, insensitive
to the genus of the scatterer, and immune from low-frequency breakdown.

\section{Generalized transmission problems}

Our approach to the Maxwell transmission problem is based on the 
analysis of two nonstandard boundary value problems governed
by the vector Helmholtz equation.
More precisely, we introduce the electric and magnetic transmission
problems as follows:

\begin{definition}
By the {\bf vector electric transmission problem} we mean the calculation of
a vector field
\[ 
\begin{array}{ll}
\bE \in C^2(D)\cap C(\overline{D}) & \mbox{for $\xb \in D$}; \\
\bE_0 \in C^2(\mathbb{R}^3 \setminus \overline{D})\cap C(\mathbb{R}^3/D) 
& \mbox{for $\xb \in \mathbb{R}^3\setminus \overline{D}$}, \end{array}
\]
with
\begin{equation}
\begin{aligned}
\Delta \bE+k^2\bE &=0, &\quad  \bx &\in D\\
\Delta \bE_0+k_0^2\bE_0 &=0,  &\quad \bx &\in \mathbb{R}^3/\overline{D}
\end{aligned}
\end{equation}
that satisfies the interface conditions
\begin{equation}\label{cond_vE1}
\bn\times\big(\bE_0-\bE\big)=\mathbf{f} \\
\end{equation}
\begin{equation}\label{cond_vE2}
\bn\times\left(\frac{\nabla\times\bE_0}{\mu_0}-
\frac{\nabla\times\bE}{\mu}\right)=\mathbf{g}
\end{equation}
\begin{equation}\label{cond_sE1}
\nabla\cdot\bE_0-\nabla\cdot\bE=q
\end{equation}
\begin{equation}\label{cond_sE2}
\bn\cdot\left(\epsilon_0 \bE_0- \epsilon \bE\right)=p
\end{equation}
and the radiation condition
\begin{equation}\label{radE}
\begin{array}{ll}
\nabla\times \bE_0(\bx)\times\frac{\bx}{|\bx|}+
\frac{\bx}{|\bx|}\nabla\cdot \bE_0(\bx)-ik_0\bE_0(\bx)=
o\Big(\frac{1}{|\bx|}\Big), \quad &|\bx|\rightarrow \infty\, ,
\end{array}
\end{equation}
where 
$\mathbf{f}\in C^{0,\alpha}_t(Div,\partial D)$, 
$\mathbf{g}\in C^{0,\alpha}_t(\partial D)$ and 
$q,p\in C^{0,\alpha}(\partial D)$.
\end{definition}

\begin{definition}
By the {\bf vector magnetic transmission problem} we mean the calculation of
a vector field
\[ 
\begin{array}{ll}
\bH \in C^2(D)\cap C(\overline{D}) & \mbox{for $\xb \in D$}; \\
\bH_0 \in C^2(\mathbb{R}^3 \setminus \overline{D})\cap C(\mathbb{R}^3/D) 
& \mbox{for $\xb \in \mathbb{R}^3\setminus \overline{D}$}, \end{array}
\]
with
\begin{equation}
\begin{aligned}
\Delta \bH+k^2\bH &=0, &\quad  \bx &\in D\\
\Delta \bH_0+k_0^2\bH_0 &=0,  &\quad \bx &\in \mathbb{R}^3/\overline{D}
\end{aligned}
\end{equation}
that satisfies the interface conditions
\begin{equation}\label{cond_vH1}
\bn\times\left(\bH_0-\bH\right)=\mathbf{f}'
\end{equation}
\begin{equation}\label{cond_vH2}
\bn\times\left(\frac{\nabla\times\bH_0}{\epsilon_0}-
\frac{\nabla\times\bH}{\epsilon}\right)= \mathbf{g}'
\end{equation}
\begin{equation}\label{cond_sH1}
\nabla\cdot\bH_0-\nabla\cdot\bH=q'
\end{equation}
\begin{equation}\label{cond_sH2}
\bn\cdot\left(\mu_0 \bH_0- \mu \bH\right)=p'
\end{equation}
and the radiation condition
\begin{equation}\label{radH}
\begin{array}{ll}
\nabla\times \bH_0(\bx)\times\frac{\bx}{|\bx|}+
\frac{\bx}{|\bx|}\nabla\cdot \bH_0(\bx)-ik\bH_0(\bx)=
o\Big(\frac{1}{|\bx|}\Big), \quad &|\bx|\rightarrow \infty\, ,
\end{array}
\end{equation}
where $\mathbf{f}'\in C^{0,\alpha}_t(Div,\partial D)$, $\mathbf{g}'\in C^{0,\alpha}_t(\partial D)$ and $q',p'\in C^{0,\alpha}(\partial D)$.
\end{definition}

\begin{thm} \label{uniquethm}
The vector electric and magnetic transmission problems have unique solutions
for $\omega \ge0$.
\end{thm}

\begin{proof}
See Appendix \ref{pdeuniq}.
\end{proof}

\begin{definition}
The layer potentials in \eqref{Skvecdef} and \eqref{Skscalardef}
are referred to as single layer potentials, with vector or scalar
densities, $\J$ and $\rho$, respectively. 
Letting $\xb'$ denote a point on $\partial D$, 
they are continuous across the interface. Their normal derivatives,
denoted by $S_k'$, satisfy the jump conditions \cite{CK1,CK2}:
\begin{equation}
\lim_{{\xb \rightarrow \xb'}^{\pm}}  
\bn(\xb') \cdot \nabla S_k[\J](\xb) = \pm\frac{1}{2} \J(\xb') + S_k'[\J](\xb') 
\label{Skdef1}
\end{equation}
\begin{equation}
 \lim_{{\xb \rightarrow \xb'}^{\pm}}  
\bn(\xb') \cdot \nabla S_k[\rho](\xb) = \pm\frac{1}{2} \rho(\xb') + S_k'[\rho](\xb') 
\label{Skdef2}
\end{equation}
where $S_k'[\J](\xb')$ and $S_k'[\rho](\xb')$ are defined in the principal value sense, 
$\lim_{{\xb \rightarrow \xb'}^-}$ denotes the interior limit ($\xb \in D$) and
$\lim_{{\xb \rightarrow \xb'}^+}$ denotes the exterior limit 
($\xb \in \mathbb{R}^3/\overline{D}$).

The double layer potential is defined by 
\begin{equation}
D_k[\rho](\xb) \equiv \int_{\partial D} \nabla g_k(\xb-\yb) \cdot \bn(\yb) \, 
\rho(\yb) \, ds_{\yb} \, .
\label{Dkdef}
\end{equation}
It is well-known to satisfy the jump condition
\[  \lim_{{\xb \rightarrow \xb'}^{\pm}}  
D_k[\rho](\xb) = \mp\frac{1}{2} \rho(\xb') + D_k[\rho](\xb') 
\]
where $D_k[\rho](\xb')$ is defined in the principal value sense.
Finally, we let $M_k$ denote the operator obtained by taking the limit:
\begin{equation}
\lim_{{\xb \rightarrow \xb'}^\pm} \bn(\xb') \times \nabla \times S_k[\J](\xb) \equiv
\mp \J(\xb') + M_k[\J](\xb') \, ,
\label{Mkdef}
\end{equation}
where $M_k[\J](\xb')$ is defined in the principal value sense.
\end{definition}

\begin{thm}
Let
\begin{equation}
\begin{aligned}\label{eerepresentation}
\bE_0&=\mu_0\nabla\times S_{k_0}[\ba]-\mu_0S_{k_0}[\bn \sigma]+
\mu_0\epsilon_0S_{k_0}[\bb]+\nabla S_{k_0}[\rho] & 
\bx\in \mathbb{R}^3/\overline{D}\\
\bE&=\mu\nabla\times S_{k}[\ba]-\mu S_{k}[\bn \sigma]+
\mu\epsilon S_{k}[\bb]+\nabla S_{k}[\rho]& \bx\in D \, .
\end{aligned}
\end{equation}
Imposing the conditions \eqref{cond_vE1} - \eqref{cond_sE2}
yields the second kind Fredholm equation
{\small
\begin{equation}
\begin{aligned}
\frac{\mu_0+\mu}{2}\ba+\big(\mu_0M_{k_0}-\mu M_{k}\big)[\ba]-
\bn\times\big(\mu_0 S_{k_0}-\mu S_{k}\big)[\bn\sigma] +\ \ &  \\
\bn\times\big( \mu_0\epsilon_0S_{k_0}-\mu\epsilon S_{k}\big)[\bb]+
\bn\times\nabla\big(S_{k_0}-S_k\big)[\rho] &= \mathbf{f}\\
\frac{\mu_0+\mu}{2}\sigma+\big(\mu_0D_{k_0}-\mu D_{k}\big)[\sigma]+
\nabla\cdot\big(\mu_0\epsilon_0 S_{k_0}-\mu\epsilon S_{k}\big)[\bb] \ \ & \\
- \omega^2\big(\mu_0\epsilon_0 S_{k_0}-\mu\epsilon S_{k}\big)[\rho] &= q\\
\frac{\epsilon_0+\epsilon}{2}\bb+ (\epsilon_0M_{k_0}-\epsilon M_{k})[\bb]
+ \bn\times\nabla\times\nabla\times\big(S_{k_0}-S_k\big)[\ba] - \ \ & \\
\bn\times\nabla\times\big(S_{k_0}-S_k\big)[\bn\sigma] &= \mathbf{g}\\
-\frac{\epsilon_0+\epsilon}{2}\rho+
(\epsilon_0S'_{k_0}-\epsilon S'_{k})[\rho]+
\bn\cdot\nabla\times\big(\epsilon_0\mu_0S_{k_0}-\epsilon\mu S_{k}\big)[\ba]
- \ \ & \\
\bn\cdot\big(\epsilon_0\mu_0S_{k_0}-\epsilon\mu S_{k}\big)[\bn\sigma] + 
\bn\cdot\big(\mu_0\epsilon_0^2 S_{k_0}-\mu\epsilon^2 S_{k}\big)[\bb]
&= p \, .
\end{aligned}
\label{veteq}
\end{equation}
}
which is invertible for
$(\mathbf{f},p,\mathbf{g},q)$ in the function space
$$C_t^{0,\alpha}(Div,\partial D)\times C^{0,\alpha}(\partial D)\times 
C_t^{0,\alpha}(\partial D)\times C^{0,\alpha}(\partial D)$$, 
so long as $\epsilon,\mu$ satisfy the conditions \eqref{microwave_cond}.
Given the solution to \eqref{veteq}, the functions given by 
\eqref{eerepresentation} solve the vector electric transmission problem.
\end{thm}

\begin{proof}
See Appendix \ref{pdeexist}.
\end{proof}

The analogous result follows trivially for the vector magnetic transmission problem.

\begin{thm}
Let
\begin{equation}
\begin{aligned}\label{Hrepresentation}
\bH_0&=\epsilon_0\nabla\times S_{k_0}(\ba)
-\epsilon_0S_{k_0}(\bn \sigma )+\mu_0\epsilon_0S_{k_0}(\bb)+\nabla S_{k_0}(\rho) & \bx\in \mathbb{R}^3/\overline{D}\\
\bH&=\epsilon\nabla\times S_{k}(\ba)-\epsilon S_{k}(\bn \sigma )+\mu\epsilon S_{k}(\bb)+\nabla S_{k}(\rho)& \bx\in D \, .
\end{aligned}
\end{equation}
Imposing the conditions \eqref{cond_vH1} - \eqref{cond_sH2}
yields a second kind Fredholm equation, identical to \eqref{veteq}
with $\{ \epsilon,\mu \}$ and 
$\{ \epsilon_0,\mu_0 \}$ interchanged.
Given the solution to this dual integral equation,
the functions given by 
\eqref{Hrepresentation} solve the vector magnetic transmission problem.
\end{thm}

\begin{thm}
Let $\bE^{in}, \bH^{in}$ denote an incoming electromagnetic field,
let $\bE_0,\bE$ denote the solution of the vector 
electric transmission problem with right hand side
\begin{equation}
\begin{aligned}
\mathbf{f}&=-\bn\times\bE^{in}\\
\mathbf{g}&=-i\omega\bn\times\bH^{in}\\
q&=0\\
p&=-\bn\cdot\epsilon_0\bE^{in} \, ,
\end{aligned}
\end{equation}
and let $\bH_0,\bH$ denote the solution of the vector
magnetic transmission problem with right hand side
\begin{equation}
\begin{aligned}
\mathbf{f}'&=-\bn\times\bH^{in}\\
\mathbf{g}'&=-i\omega\bn\times\bE^{in}\\
q'&=0\\
p'&=-\bn\cdot\mu_0\bH^{in} \, .
\end{aligned}
\end{equation}
Then, the fields $\bE_0,\bE,\bH_0,\bH$ satisfy the Maxwell equations
and solve the Maxwell transmission problem.
\end{thm}

\begin{proof}
The result follows from Theorem \ref{uniquethm} and the fact that
the desired scattered fields satisfy the vector electric and magnetic
transmission problems by inspection. 
The boundary data $\mathbf{f},\mathbf{g},q,p$  and 
$\mathbf{f}',\mathbf{g}',q',p'$
satisfy the required regularity conditions assuming that the incoming
fields
$\bE^{in}$ and $\bH^{in}$ are induced by exterior sources away from 
$\partial D$.
\end{proof}

\begin{thm} \label{stabthm}
The solution of the vector electric and magnetic transmission problems
satisfy the following stability properties uniformly on the interval 
$\omega\in [0,\omega_{\max}]$.  
\begin{equation}\label{ineq_1}
\|\bE_0\|_{\alpha,\mathbb{R}^3\backslash D}\le K(\partial D,k_{\max})
\Big(\|\mathbf{f}\|_{\alpha,\partial D}+\|\mathbf{g}\|_{\alpha,\partial D}+\|q\|_{\alpha,\partial D}+\|p\|_{\alpha,\partial D}\Big)
\end{equation}
\begin{equation}\label{ineq_2}
\|\bH_0\|_{\alpha,\mathbb{R}^3\backslash D}\le K(\partial D,k_{\max})
\Big(\|\mathbf{f}\|_{\alpha,\partial D}+\|\mathbf{g}\|_{\alpha,\partial D}+\|q\|_{\alpha,\partial D}+\|p\|_{\alpha,\partial D}\Big)
\end{equation}
For Maxwellian incoming fields, we have 
\begin{equation}\label{ineq_3}
\|\bE_0\|_{\alpha,\mathbb{R}^3\backslash D}\le K(\partial D,k_{\max})
\Big(\|\bE^{in}\|_{\alpha,\partial D}+\|\bH^{in}\|_{\alpha,\partial D}\Big)
\end{equation}
\begin{equation}\label{ineq_4}
\|\bH_0\|_{\alpha,\mathbb{R}^3\backslash D}\le K(\partial D,k_{\max})
\Big(\|\bE^{in}\|_{\alpha,\partial D}+\|\bH^{in}\|_{\alpha,\partial D}\Big).
\end{equation}
\end{thm}

\begin{proof}
See Appendix \ref{pdestab}.
\end{proof}

As a consequence of the preceding theorems, one can solve the Maxwell transmission 
problem replacing it with the vector electric and magnetic transmission problems.
This permits evaluation of the fields all the way to $\omega = 0$ without low 
frequency breakdown and without regard to the genus of the surface.

\section{High Frequency Scaling}

While the representation 
\eqref{eerepresentation} is sufficient for proving existence and uniqueness,
it is not well scaled at high frequencies. For the sake of simplicity,
we will asume that the scatterer has dimensions on the order of unity,
so that $\omega$ itself is a measure of size in terms of wavelength.
Following Kress \cite{Kress_sphere} and our earlier work \cite{ourSphere}, 
when $\omega > 1$, we suggest that the  representation for the electromagnetic
field be modified as follows:

\begin{equation}
\begin{aligned}\label{representation}
\bE_0&=\mu_0\nabla\times S_{k_0}(\ba)-\omega\mu_0S_{k_0}(\bn \sigma )+\omega\mu_0\epsilon_0S_{k_0}(\bb)+\nabla S_{k_0}(\rho) & \bx\in \mathbb{R}^3/\overline{D}\\
\bE&=\mu\nabla\times S_{k}(\ba)-\omega\mu S_{k}(\bn \sigma )+\omega\mu\epsilon S_{k}(\bb)+\nabla S_{k}(\rho)& \bx\in D
\end{aligned}
\end{equation}

We also rescale the jump conditions 
\eqref{cond_vE1}-\eqref{cond_sE2} as follows:
\begin{equation*}  
\bn\times\big(\bE_0-\bE\big)=\mathbf{f}=-\bn\times\bE^{\In}\\
\end{equation*}
\begin{equation*}  
\frac{1}{\omega}\bn\times\Big(\frac{\nabla\times\bE_0}{\mu_0}-\frac{\nabla\times\bE}{\mu}\Big)=\frac{1}{\omega}\mathbf{g}=-\frac{1}{\omega}\bn\times\frac{\nabla\times\bE^{\In}}{\mu_0}
\end{equation*}
\begin{equation*}  
\frac{1}{\omega}\big(\nabla\cdot\bE_0-\nabla\cdot\bE=q\big)=-\frac{1}{\omega}\nabla\cdot\bE^{\In}
\end{equation*}
\begin{equation*}  
\bn\cdot\big(\epsilon_0 \bE_0- \epsilon \bE\big)=p=-\bn\cdot\epsilon_0\bE^{\In}
\end{equation*}

This results in the following (rescaled) 
system of equations:
\begin{equation}
\mathbf{(B_s+K_s)x=y_s},
\end{equation}
where
\begin{equation}\label{Op_Bs}
\mathbf{B_s}:=\left( \begin{array}{cccc}
I_1 & \omega B_{12} & \omega B_{13} & 0 \\
0 & I_2 & B_{23} & 0 \\
0 & 0 & I_3 & 0 \\
B_{41} & 0 & 0 & I_4 
\end{array} \right)
\end{equation}

\begin{equation}\label{Op_Ks}
\mathbf{K_s}:=\left( \begin{array}{cccc}
K_{11} & 0 & 0 & K_{14} \\
0 & K_{22} & 0 & \frac{K_{24}}{\omega } \\
\frac{K_{31}}{\omega } & K_{32} & K_{33} & 0 \\
0 & \omega K_{42} & \omega K_{43} & K_{44}
\end{array} \right)
\end{equation}
\begin{equation}
\mathbf{x}:=\left( \begin{array}{c}
\ba \\  \sigma \\ \bb \\ \rho
\end{array} \right); \ \ \mathbf{y_s}:=
\left( \begin{array}{c}
\mathbf{f} \\  \frac{q}{\omega } \\ \frac{\mathbf{g}}{\omega } \\ p
\end{array} \right)
\end{equation}

\section{Condition number analysis}
\label{cond_section}

To illustrate the behavior of the various methods discussed above,
we implemented all of the integral operators for a spherical
scatterer, expanding
each surface current in vector spherical harmonics and each 
charge density in scalar spherical harmonics, as in \cite{ourSphere}.
From this it is straightforward to compute the condition number of
the various linear systems of interest for any 
$\epsilon$ and $\mu$, where we assume the exterior
permeability and permittivity are normalized to $\epsilon_0=1,\mu_0=1$. 

In our first experiments, we plot
the condition number of our decoupled field integral equation (DFIE),
the decoupled charge-current formulation (based on the M\"{u}ller
integral equation), and a standard charge-current formulation as a function 
of angular frequency $\omega$ (Fig. \ref{Conditionfig}).
The precise charge-current formulation that we use is obtained from
the standard representation for the fields in terms of potentials and 
antipotentials (\ref{Eeval}-\ref{Heval}), but imposing the 
continuity condition \eqref{rhodef}
in the form
\[ \nabla\cdot S_{k_0} (\J) -i\omega S_{k_0}(\rho)=0.  \]
With
\begin{eqnarray*}
\E_0 = -\nabla\times S_{k_0}(\K)+i\omega \mu_0 S_{k_0}(\J)-\nabla S_{k_0}(\rho),
\nonumber
\end{eqnarray*}
this is accomplished in a weak sense by simply replacing
the normal component of $\E_0$ with

{\small 
\begin{eqnarray*}
\bn\cdot\E_0 = -\bn\cdot\nabla\times S_{k_0}(\K)+i\omega \mu_0 \bn\cdot S_{k_0}(\J)-\bn\cdot \nabla S_{k_0}(\rho)+\eta\left(\nabla\cdot S_{k_0} (\J) -i\omega S_{k_0}(\rho)\right),
\end{eqnarray*}
}
where $\eta$ is an arbitrary parameter that defines a family of numerical
methods. 

Even for a fixed geometry, the space of possible integral equations
is high-dimensional, depending on $\omega$, $\epsilon$, and $\mu$ 
(and $\eta$ for the charge-current formulation).
Thus, Figs. \ref{Conditionfig} and \ref{Conditionfig2} 
only shows a sample of the possible behaviors
of the various methods.
In each plot, we fix $\epsilon$ and $\mu$ as indicated 
and scan the frequency $\omega$ in the range $[0,10]$. 

In Fig. \ref{Conditionfig}, 
the left column corresponds to settng $\epsilon = -2+i$, $\mu = -1 + i$,
with $\eta = 0,1,i$ in the three rows, respectively. 
Near resonances only appear to occur for $\eta = 0$. 
However, note in Fig. \ref{Conditionfig2} (left), that for 
$\epsilon = -1+i$, $\mu = 1$, it is the charge-current formulation with 
$\eta = 1$ that is badly behaved (left).
For the right-hand side of Fig. \ref{Conditionfig}, 
we searched for values of 
$\epsilon$, $\mu$ where the decoupled charge-current blows up.
For $\epsilon = 1+i$, $\mu = 1$, the charge-current formulation with 
$\eta = 0$ appears to behave well (Fig. \ref{Conditionfig2}, center).
Note, finally, that
in the absence of dissipation (when $\epsilon$ and $\mu$ are
real), the decoupled charge-current formulation appears to be the best behaved
(Fig. \ref{Conditionfig2}, right). 

\begin{figure}[H]
\begin{center}
\includegraphics[width=5.6in]{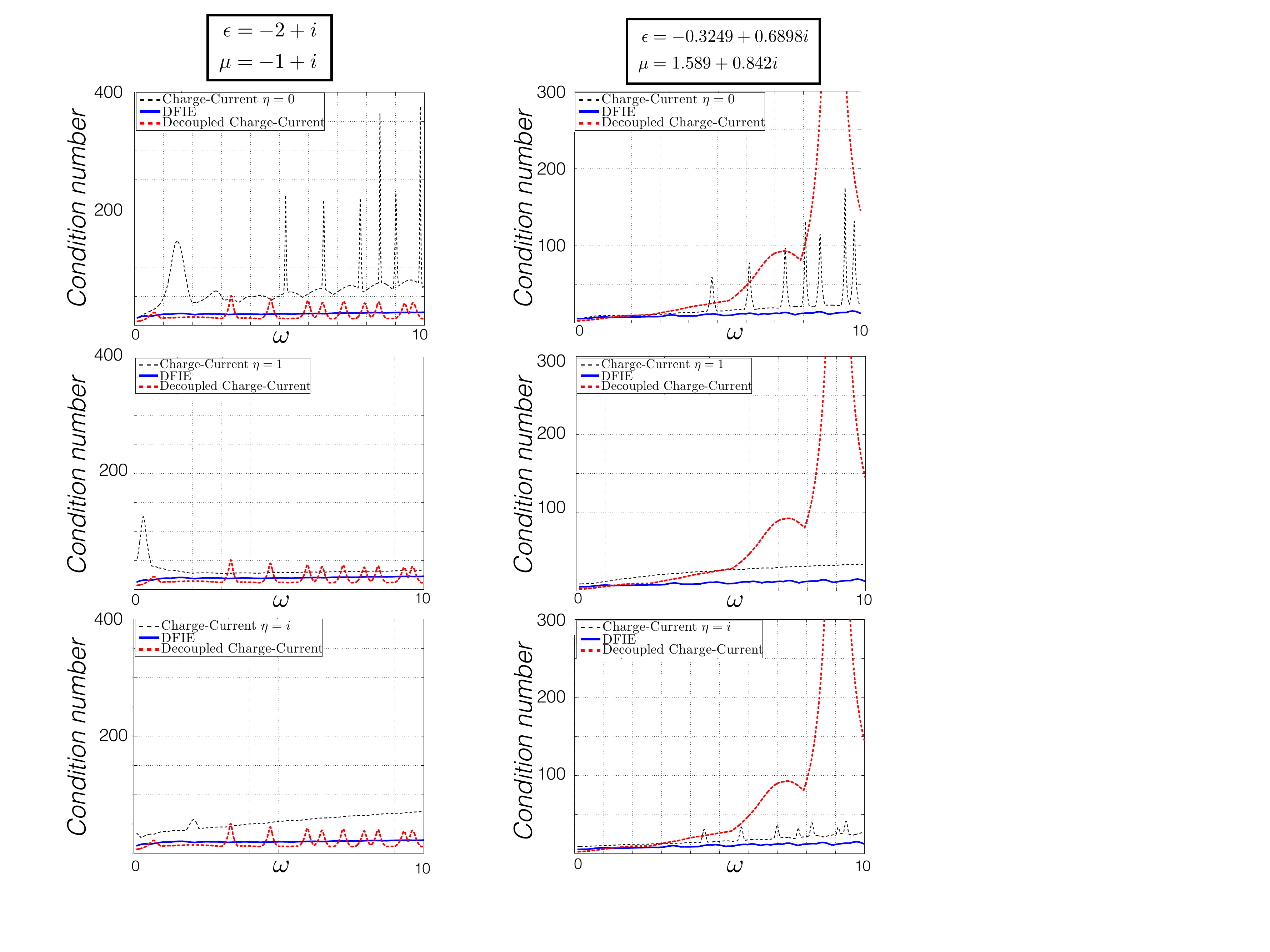}%
\end{center}
\caption{Condition number of the (scaled) DFIE, the
charge-current formulation with $\eta = 0,1,i$, 
and the decoupled charge-current (M\"{u}ller)
formulation for a homogeneous dielectric sphere of unit radius as a function
of frequency $\omega$ for $\epsilon= -2+i,\mu = -1+i$ and for 
$\epsilon= -0.3249 + 0.6898i,\mu = 1.589+0.842i$.
\label{Conditionfig}}
\end{figure}

\begin{figure}[H]
\begin{center}
\includegraphics[width=5.6in]{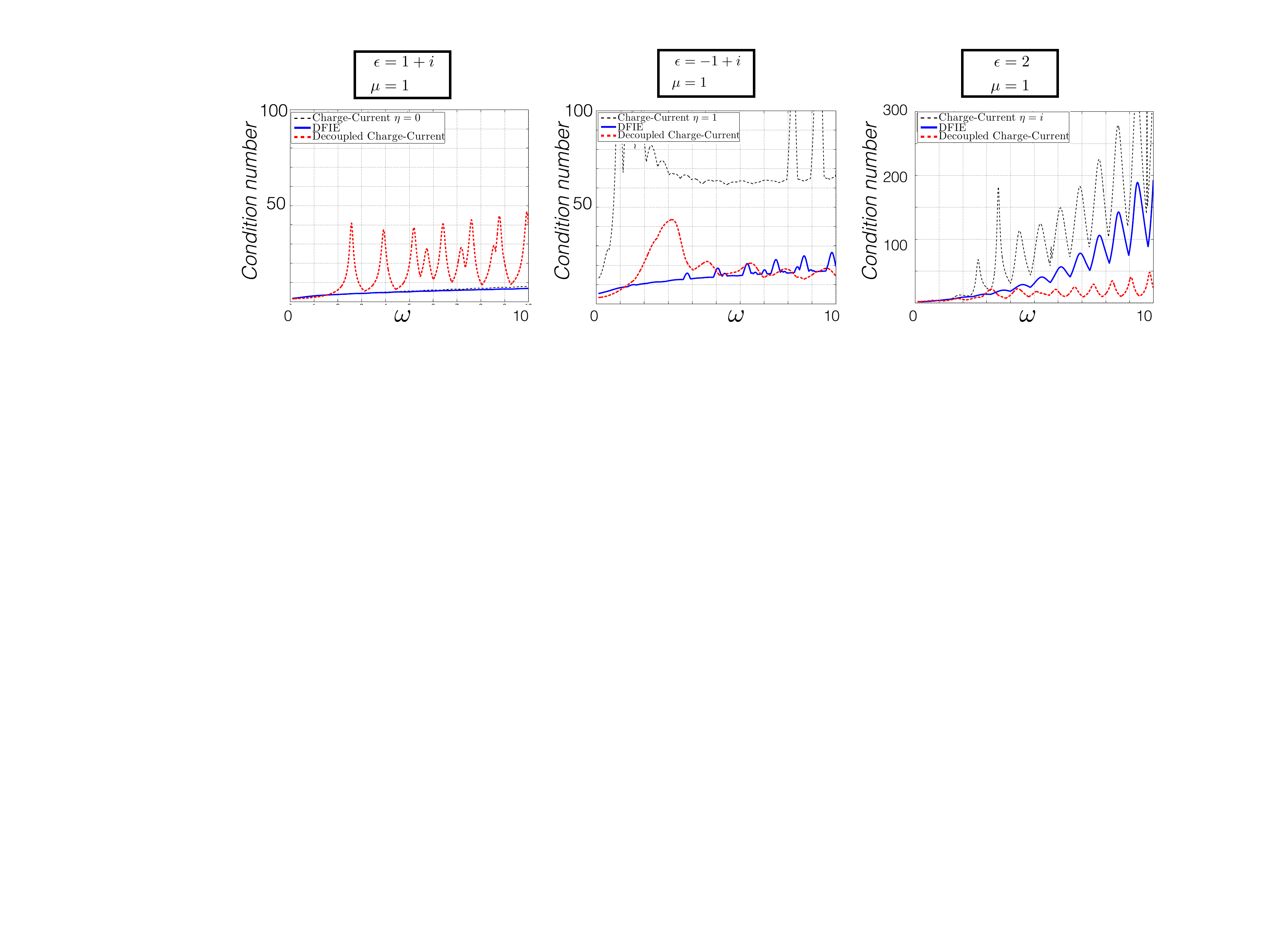}%
\end{center}
\caption{Condition number of the (scaled) DFIE, the charge-current
formulation with $\eta = 0$, and the decoupled charge-current (M\"{u}ller)
formulation for a homogeneous dielectric sphere of unit radius as a function
of frequency $\omega$ for various material parameters.
\label{Conditionfig2}}
\end{figure}

We plot the locations of 
spurious resonances of the decoupled charge-current formulation 
in Fig. \ref{ccresonances}.
For each point in the plane defined by the real parts of $\epsilon$ and 
$\mu$, we use the scheme described in 
\cite{homogeneous_uniqueness_10} to find positive
imaginary components for $\epsilon$ and $\mu$ that lead to
blow-up of the integral equation in the range $\omega \in [0,10]$.
The material is lossy (dissipative) by construction, so these resonances
are non-physical.

\begin{figure}[H]
\begin{center}
\includegraphics[width=4in]{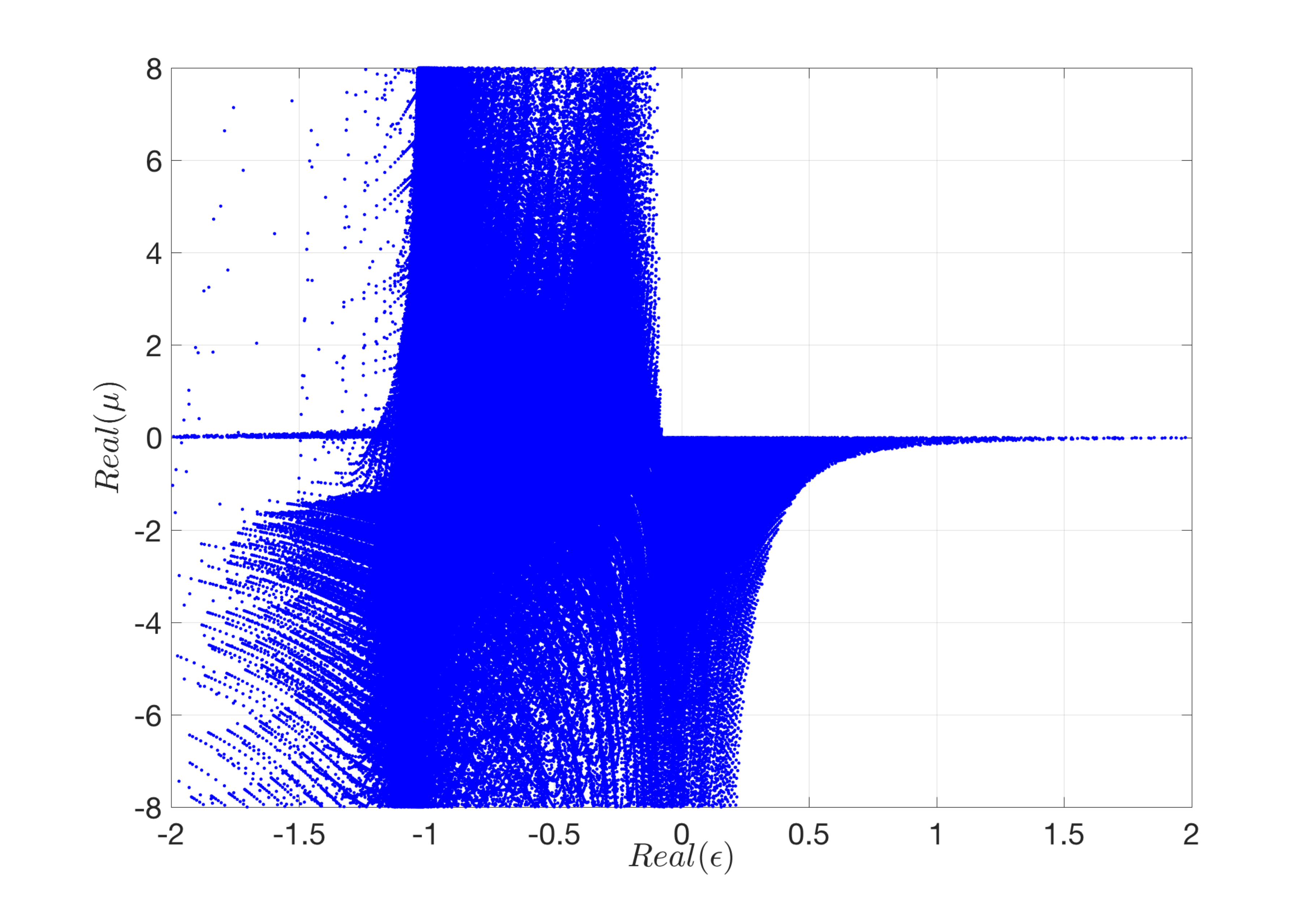}%
\end{center}
\caption{
Spurious resonances in the decoupled charge-current formulation.}
\label{ccresonances}
\end{figure}

\section{Further numerical validation}

While the primary purpose of this paper
is to present the derivation and analysis of the decoupled field integral 
equation, we illustrate its performance here using a high-order
locally corrected Nystr\"om discretization \cite{Nystrom_our}. 
Without entering into details, this method uses $n_{nodes} = 7$ points 
per curved triangular element for 5th order accuracy, 
$n_{nodes} = 25$ points per curved triangular element for 10th order accuracy, 
and $n_{nodes} = 45$ nodes per curved triangular element for 14th order 
accuracy.
The total number of unknowns is $N=6 \, N_{tri} \, n_{nodes}$,
as there are two tangential vector fields and two scalar unknowns.
corresponding to six degrees of freedom, at each point.

For our first example, we consider
the obstacle to be a single spheroid centered at the origin with 
semi-principal axes of length $a=1, b=2$, and $c=3$,
with $\epsilon=1.5$ and $\mu=1$, $\epsilon_0 = \mu_0 = 1$  and $\omega = 1$.
The incoming wave is assumed to be a plane wave propagating in the 
$z$-direction.  In Fig. \ref{Error_1} we plot the estimated 
relative error for the interior and 
exterior regions, using a reference solution with 200 triangles and 45 nodes 
per triangle.

\begin{figure}[H]
\begin{center}
\includegraphics[width=4in]{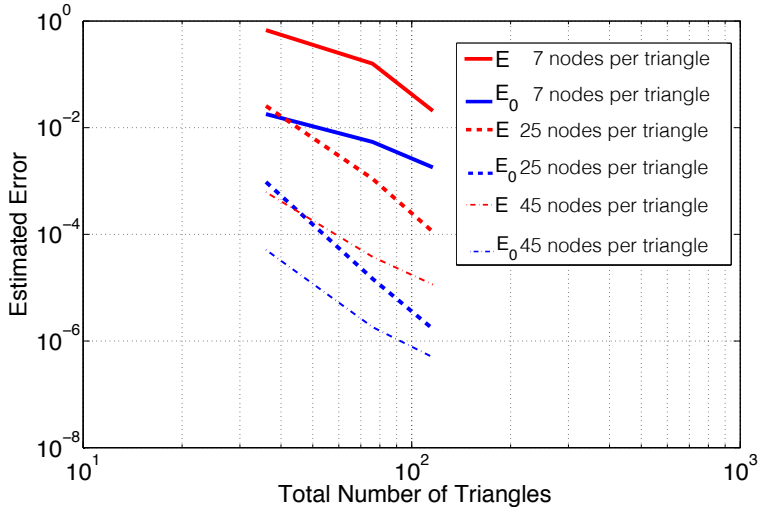}%
\end{center}
\caption{Estimated error for an ellipsoid with semi-principal axes of length
$a=1, b=2, c=3$ and $\epsilon=1.5,\mu=1$. 
The }
\label{Error_1}
\end{figure}

For our second example, we consider the same 
ellipsoid with $\omega = 1$ but with $\epsilon=-3+i$ and $\mu=-2+.5i$,
a so-called double negative index material.
Fig. \ref{Error_2} shows the estimated relative error for the interior and 
exterior regions.

\begin{figure}[H]
\begin{center}
\includegraphics[width=4in]{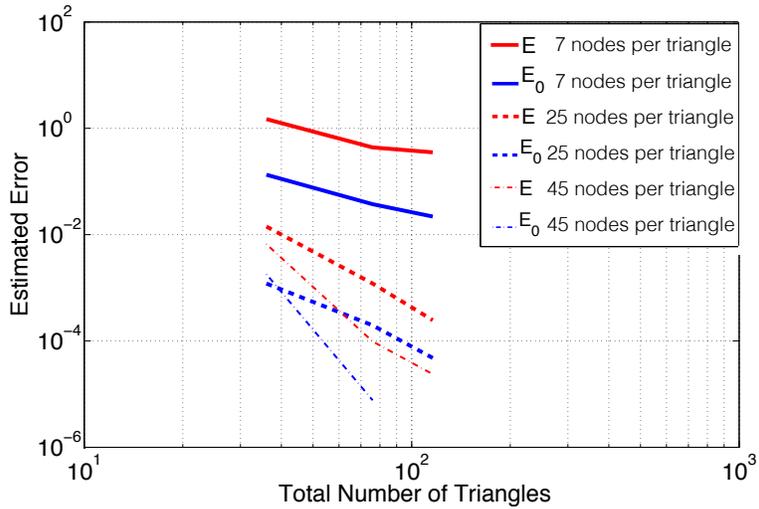}%
\end{center}
\caption{Estimated error for an ellipsoid with semi-principal axes of length
$a=1, b=2, c=3$ and $\epsilon=-3+i,\mu=-2+.5i$.}
\label{Error_2}
\end{figure}

Unlike the naive implementation of M\"{u}ller's method (or the PMCHW scheme
\cite{PMCHW}), there is clearly no ``dense-mesh breakdown" in evaluating 
the electromagnetic field, cconsistent with the theory.

To further demonstrate the stable behavior of our scheme at low frequencies,
we consider ithe obstacle to be a sphere of radius $R=1$ with 
$\epsilon=1.3, \mu=1$.
Fig. \ref{Error_3} shows the error, with the exact solution computed via the 
the Mie solution. There is clearly no low frequency breakdown.

\begin{figure}[H]
\begin{center}
\includegraphics[width=4in]{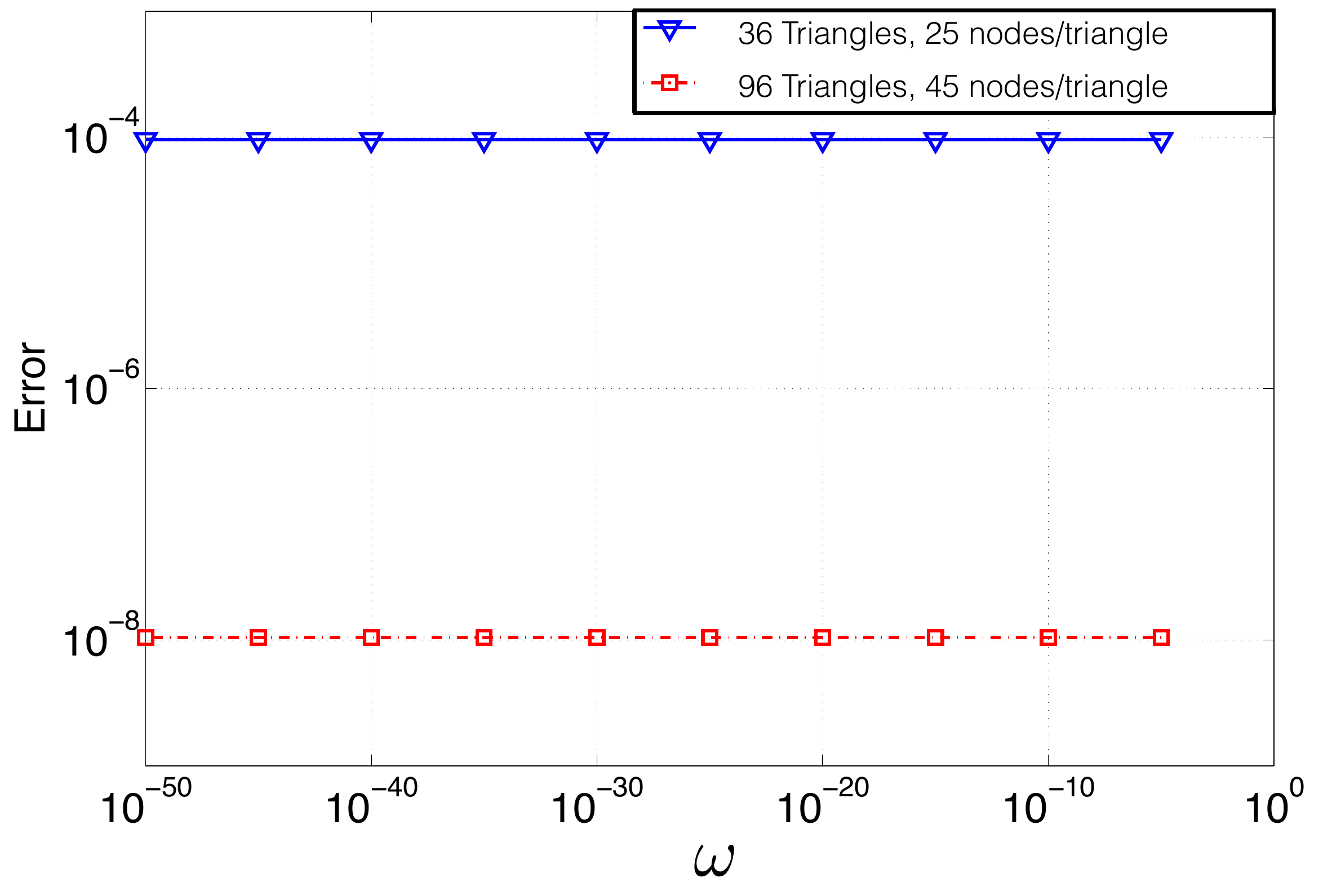}%
\end{center}
\caption{Estimated error for a sphere of radius $R=1$ with 
$\epsilon=1.3, \mu=1$.}
\label{Error_3}
\end{figure}

Our final example is a superellipsoid: $x^8+y^8+z^8\le 1$ with  
$\epsilon=1.3$, $\mu=1$, and $\omega = 0.1$. 
Fig. \ref{Error_4} shows the estimated relative error for the interior
and exterior regions. The error in the exterior region is smaller since it 
is measured at a distance $R=10$, where the integrals involve smoother integrands.
Notice, again, the absence of dense-mesh breakdown.

\begin{figure}[H]
\begin{center}
\includegraphics[width=4in]{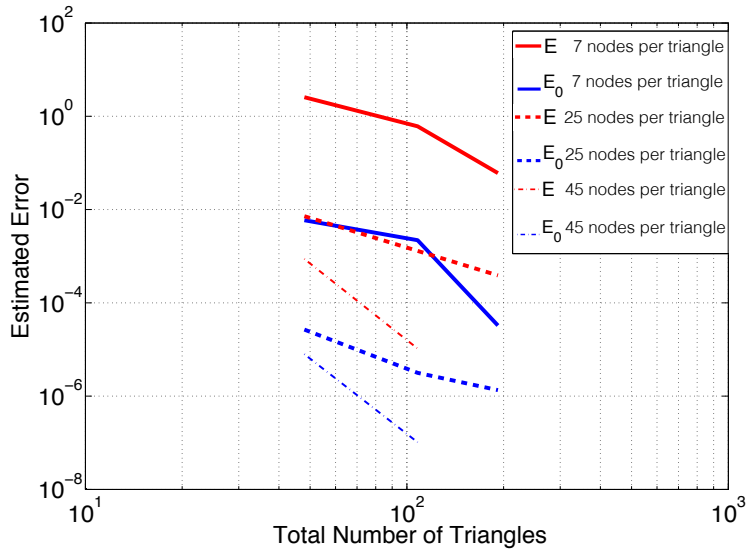}%
\end{center}
\caption{Estimated error for a superellipsoid $x^8+y^8+z^8\le1$ with 
$\epsilon=1.3, \mu=1$.}
\label{Error_4}
\end{figure}

\section{Conclusions}

We have presented a new method for simulating
electromagnetic scattering from homogeneous penetrable bodies,
based on reformulating the 
Maxwell equations in terms of two uncoupled vector Helmholtz systems.
one for the electric feld and one for the magnetic field.
We have shown that these partial differential equations have unique 
solutions and that those solutions correspond to the desired electromagnetic
field, Furthermore, we have shown that the vector Helmholtz equations
can be recast as boundary integral equations which are well-conditioned and
resonance-free for all lossy materials, including doubly 
negative materials, where $\Re{(\epsilon)}<0,\Re{(\mu)}<0$. 
We refer to these as {\em decoupled field integral equations}. They are
insensitive to the genus of the scatterers, and immune from dense-mesh 
(low frequency) breakdown.

Previously developed charge-current formulations avoid dense-mesh breakdown 
but can be subject to spurious resonances and can have erratic behavior, 
as seen in our numerical experiments above. 
Our approach is based on Fredholm integral equations of the second
kind, equipped with relatively straightforward proofs of existence 
and uniqueness based on standard energy estimates and the Rellich Lemma.

The extension of our method to problems involving both perfect conductors
and dielectrics is underway, as is its implementation using suitable fast
algorithms. It will be interesting to compare its 
performance with the generalized
Debye method - the only other approach which has been proven to lead to
well-posed integral equations for all passive materials  
\cite{gDEBYE2}.  Results from these developments will be reported at a later
date.

\section*{Acknowledgments}
This work was supported in part 
by the Office of the Assistant Secretary of Defense for Research and
Engineering and AFOSR under NSSEFF Program Award FA9550-10-1-0180 and in
part by the Spanish Ministry of Science and Innovation under the 
project TEC2016-78028-C3-3-P.

\appendix 
\section{Uniqueness for the Maxwell transmission problem} \label{maxuniq}

\begin{thm}\label{thm_uni_vector}
The Maxwell transmission problem has a unique solution for any permittivity
and permeability satisfying conditions (\ref{microwave_cond}) for $\omega>0$.
\end{thm}
\begin{proof}
Suppose that $\bE,\bH,\bE_0,\bH_0$ denotes a solution to the Maxwell transmission
problem with Helmholtz parameters 
$k^2=\omega^2\epsilon\mu, k_0^2=\omega^2\epsilon_0\mu_0$,
and homogeneous boundary data, so that
\begin{equation}
\begin{aligned}
\bn\times\big(\bE_0-\bE\big)=0\\
\bn\times\Big(\frac{\nabla\times\bE_0}{\mu_0}-\frac{\nabla\times\bE}{\mu}\Big)=0.
\end{aligned}
\end{equation} 

Due to the regularity properties of the fields at the boundary, we may apply 
the Rellich lemma \cite{CK1} which states that uniqueness holds if 
$\Im{E} \geq 0$ where 
\begin{equation}
I_v=k_0\int_{\partial D}\big(\bn\times\bE_0\cdot\nabla\times\overline{\bE_0}\big)ds
\label{Edef}
\end{equation}
and $\Im{E}$ denotes its imaginary part. 
Using the jump conditions and Green's identity \cite{CK1} we obtain:
\begin{equation}
\begin{aligned}
I_v&=k_0\int_{\partial D}\big(\bn\times\bE_0\cdot\nabla\times\overline{\bE_0}\big)ds=k_0\overline{\Big(\frac{\mu_0}{\mu}\Big)}\int_{\partial D}\big(\bn\times\bE\cdot\nabla\times\overline{\bE}\big)ds=\\
&=k_0\overline{\Big(\frac{\mu_0}{\mu}\Big)}\int_{D}|\nabla\times\bE|^2-\overline{k}^2|\bE|^2 dv=\frac{1}{\overline{\mu}}\Big(k_0\overline{\mu_0}\int_{D}|\nabla\times\bE|^2dv\Big)+=\\
&-\overline{\epsilon}\Big(k_0\omega^2\overline{\mu_0}\int_{D}|\bE|^2 dv\Big)=\frac{1}{\overline{\mu}}\alpha-\overline{\epsilon}\beta \, .
\end{aligned}
\end{equation} 
From (\ref{microwave_cond}), we find that $\alpha$ and $\beta$ are real, 
non-negative numbers and $\Im(I_v)\ge0$. 
Thus, $\bE_0=0 \ \forall \bx \in \mathbb{R}^3/\overline{D}$. 
Since $\bH_0=\frac{\nabla\times\bE_0}{i\omega\mu_0}$, we also have that 
$\bH_0=0 \ \forall \bx \in \mathbb{R}^3/\overline{D}$. 
From the jump conditions we get $\bn\times\bE=\bn\times\bH=0$, therefore, using the representation theorem (see \cite{CK1}) we get that $\bE=0,\bH=0, \ \forall \bx \in D$.

We will have occasion to consider the dual case, where $\epsilon,\mu>0$ and 
$\epsilon_0,\mu_0$ satisfy \eqref{microwave_cond} for which we now show that
we the transmission problem also has a unique solution.
First, if $\Im(k_0)=0$, the result above is sufficient.
Thus, we may assume that $\Im(k_0>0)$. 
Let $B_R$ be a ball with radius $R$ that contains the obstacle $D$. 
Applying Green's identity we obtain
\begin{equation}
\begin{aligned}
&\frac{1}{\mu_0}\int_{\partial B_R}\big(\bn\times\overline{\bE_0}\cdot\nabla\times\bE_0\big)ds=\frac{1}{\mu_0}\int_{B_R/\overline{D}}|\nabla\times\bE_0|^2-k_0^2|\bE_0|^2 dv+\\
&\qquad\qquad +\frac{1}{\mu}\int_{D}|\nabla\times\bE|^2-k^2|\bE|^2 dv.
\end{aligned}
\end{equation} 
Taking the limit as $R\rightarrow \infty$, the left hand side tends to zero 
due to the radiation condition and the fact that $\Im(k_0>0)$ so that the 
outer material is lossy. Thus,
\begin{equation}
\begin{aligned}
0=\frac{1}{\mu_0}\int_{\mathbb{R}^3/\overline{D}}|\nabla\times\bE_0|^2-k_0^2|\bE_0|^2 dv
+\frac{1}{\mu}\int_{D}|\nabla\times\bE|^2-k^2|\bE|^2 dv.
\end{aligned}
\end{equation}
Taking the imaginary part, we get
\begin{equation}
\begin{aligned}
0=&\Im\Big(\frac{1}{\mu_0}\int_{\mathbb{R}^3/\overline{D}}|\nabla\times\bE_0|^2-k_0^2|\bE_0|^2 dv\Big)=\\
=&\Im\Big(\frac{1}{\mu_0}\Big)\int_{\mathbb{R}^3/\overline{D}}|\nabla\times\bE_0|^2 dv-\Im(\epsilon_0)\omega^2\int_{\mathbb{R}^3/\overline{D}}|\bE_0|^2 dv.
\end{aligned}
\end{equation} 
Recall now that $\mu_0$ and $\epsilon_0$ cannot both be real. 
Thus, if  $\Im(\mu_0)>0$, then $\nabla\times\bE_0=0$.
If $\Im(\epsilon_0)>0$, then $\bE_0=0$. In either case, 
we get $\bH_0=\bE_0=0 \ \forall \bx \in \mathbb{R}^3/\overline{D}$, and 
$\bE=0,\bH=0, \ \forall \bx \in D$, as desired.
\end{proof}

\begin{remark}
At zero frequency, the Maxwell transmission problem
no longer has a unique solution,
unless additional constraints are imposed on the normal data.
Our formulation in terms of the vector electric and magnetic transmission
problems is unique even at zero frequency and yields the (unique) limit of 
the Maxwell transmission problem as $\omega \rightarrow 0^+$.
\end{remark}

\section{The scalar transmission problems} \label{scalarpdes}

\begin{definition}
By the {\bf scalar electric transmission problem} we mean the calculation of
a scalar field
\[ 
u = \left\{ \begin{array}{ll}
u_i \in C^2(D)\cap C(\overline{D}) & \mbox{if $\xb \in D$}; \\
u_0 \in C^2(\mathbb{R}^3 \setminus \overline{D})\cap C(\mathbb{R}^3/D) 
& \mbox{if $\xb \in \mathbb{R}^3\setminus \overline{D}$}, \end{array}
\right.
\]
with
\begin{equation}
\begin{aligned}
\Delta u_i+k^2u_i &=0, &\quad  \bx &\in D\\
\Delta u_0+k_0^2u_0 &=0,  &\quad \bx &\in \mathbb{R}^3/\overline{D}
\end{aligned}
\end{equation}
that satisfies the interface conditions
\begin{equation}\label{cond_seH1}
u_0-u_i=q
\end{equation}
\begin{equation}\label{cond_seH2}
\frac{1}{\mu_0}\frac{\partial u_0}{\partial n}-
\frac{1}{\mu}\frac{\partial u_i}{\partial n}=p
\end{equation}
and the Sommerfeld radiation condition
\begin{equation}\label{radS}
\begin{array}{ll}
\frac{\bx}{|\bx|}\cdot\nabla u_0(\bx)-ik_0u_0(\bx)=o\Big(\frac{1}{|\bx|}\Big), \quad &|\bx|\rightarrow \infty\, ,
\end{array}
\end{equation}
where $p,q\in C^{0,\alpha}(\partial D)$.
Since we will have occasion to consider the scalar transmission problem satisfied by 
$u = \nabla \cdot \bE$, and our smoothness assumptions on $\bE$ don't guarantee
differentiability at the boundary, we will sometimes replace 
the condition \eqref{cond_seH2} 
wth its weak counterpart:
\begin{equation}
\lim_{h\rightarrow 0, h>0}\int_{\partial D}
\left(\frac{1}{\mu_0}\frac{\partial u_0}{\partial n}(\bx+h\bn)-\frac{1}{\mu}\frac{\partial u_i}{\partial n}(\bx-h\bn)-p(\bx)\right)w(\bx)ds=0
\label{weakform}
\end{equation}
for all $w(\bx)\in C^{1,\alpha}(\partial D)$.
\end{definition}

\begin{definition}
By the {\bf scalar magnetic transmission problem} we mean the calculation of
a scalar field
\[ 
v = \left\{ \begin{array}{ll}
v_i \in C^2(D)\cap C(\overline{D}) & \mbox{if $\xb \in D$}; \\
v_0 \in C^2(\mathbb{R}^3 \setminus \overline{D})\cap C(\mathbb{R}^3/D) 
& \mbox{if $\xb \in \mathbb{R}^3\setminus \overline{D}$}, \end{array}
\right.
\]
with
\begin{equation}
\begin{aligned}
\Delta v_i+k^2v_i &=0, &\quad  \bx &\in D\\
\Delta v_0+k_0^2v_0 &=0,  &\quad \bx &\in \mathbb{R}^3/\overline{D}
\end{aligned}
\end{equation}
that satisfies the interface conditions
\begin{equation}\label{cond_shH1}
v_0-\varphi=q'
\end{equation}
\begin{equation}\label{cond_shH2}
\frac{1}{\epsilon_0}\frac{\partial v_0}{\partial n}-\frac{1}{\epsilon}\frac{\partial v_i}{\partial n}=p'
\end{equation}
and the Sommerfeld radiation condition (\ref{radS}).
Since we will have occasion to consider the scalar transmission problem satisfied by 
$v = \nabla \cdot \bH$, and our smoothness assumptions on $\bH$ don't guarantee
differentiability at the boundary, we will sometimes replace 
the condition \eqref{cond_shH2} 
wth its weak counterpart, as in \eqref{weakform}.
\end{definition}

\begin{lemma}
\label{scalunique}
The scalar electric and magnetic transmission problems have unique solutions for
permeabilities and permittivities satisfying \eqref{microwave_cond}
for $\omega \geq 0$.
\end{lemma}

\begin{proof}
We restrict our attention to the electric transmission problem, since the proof for 
magnetic case is analogous. Thus,
Suppose that $\psi,\psi_0$ are a solution of the homogeneous scalar electric transmission
problem (that is, $p,q = 0$). From Theorem 3.3 in \cite{Alemanes}, 
we find that the homogeneous solution satisfies the following regularity 
condition:
$\psi_0\in C^2(\mathbb{R}^3/\overline{D})\cap C^{1,\alpha}(\mathbb{R}^3/D)$, $\psi\in C^2(D)\cap C^{1,\alpha}(\overline{D})$. 

First, let us assume that $\omega >0$,
$\epsilon_0,\mu_0 > 0$ and 
$\epsilon,\mu$ satisfy the material properties 
given by (\ref{microwave_cond}).
In order to apply the Rellich lemma we consider the quantity
\begin{equation}
\begin{aligned}
E &=k_0\int_{\partial D}\psi_0\frac{\partial\overline{\psi_0}}{\partial n}ds=k_0\overline{\Big(\frac{\mu_0}{\mu}\Big)}\int_{\partial D}\psi\frac{\partial\overline{\psi}}{\partial n}ds=k_0\overline{\Big(\frac{\mu_0}{\mu}\Big)}\int_{D}|\nabla\psi|^2-\overline{k}^2|\psi|^2dv=\\
&=\frac{1}{\overline{\mu}}\Big(k_0\overline{\mu_0}\int_{D}|\nabla\psi|^2dv\Big)-\overline{\epsilon}\Big(\omega^2k_0\overline{\mu_0}\int_{D}|\psi|^2dv\Big)=\frac{1}{\overline{\mu}}\alpha-\overline{\epsilon}\beta,
\end{aligned}
\end{equation}
where $\alpha$ and $\beta$ are real, non-negative numbers. Therefore,
$\Im(E)\ge0$ and $\psi_0=0 \ \forall \bx \in \mathbb{R}^3/\overline{D}$. From the jump conditions, we get $\psi|_{\partial D}=\frac{\partial \psi}{\partial n}=0$.
Therefore, from the representation theorem, we have that 
$\psi=0 \ \forall \bx \in D$.

Let us now consider the dual problem, where $\epsilon,\mu > 0$ and
$\epsilon_0,\mu_0$ satisfy \eqref{microwave_cond}.
If $\Im(k_0)=0$, then we are in the situation above. Thus, we can assume
that $\Im(k_0>0)$. Let $B_R$ be a ball of radius $R$ that contains the obstacle
$D$. Applying Green's identity we have
\begin{equation}
\begin{aligned}
&\frac{1}{\mu_0}\int_{\partial B_R}\overline{\psi_0}\frac{\partial\psi}{\partial n}ds=\frac{1}{\mu_0}\int_{B_R/\overline{D}}|\nabla\psi_0|^2-k_0^2|\psi_0|^2 dv+\frac{1}{\mu}\int_{D}|\nabla\psi|^2-k^2|\psi|^2 dv.
\end{aligned}
\end{equation} 
Taking the limit $R\rightarrow \infty$, the left-hand side tends to zero due 
to the radiation condition and the fact that the outer region is lossy. Thus, 
\begin{equation}
\begin{aligned}
0=\frac{1}{\mu_0}\int_{\mathbb{R}^3/\overline{D}}|\nabla\psi_0|^2-k_0^2|\psi_0|^2 dv
+\frac{1}{\mu}\int_{D}|\nabla\psi|^2-k^2|\psi|^2 dv.
\end{aligned}
\end{equation} 
Taking the imaginary part, we get
\begin{equation}
\begin{aligned}
0=&\Im\Big(\frac{1}{\mu_0}\int_{\mathbb{R}^3/\overline{D}}|\nabla\psi_0|^2-k_0^2|\psi_0|^2 dv\Big)
=\Im\Big(\frac{1}{\mu_0}\Big)\int_{\mathbb{R}^3/\overline{D}}|\nabla\psi_0|^2 dv-\Im(\epsilon_0)\omega^2\int_{\mathbb{R}^3/\overline{D}}|\psi_0|^2 dv.
\end{aligned}
\end{equation} 
Now, $\mu_0$ and $\epsilon_0$ cannot both be real. 
If $\Im(\mu_0)>0$, then $\nabla\psi_0=0$. If $\Im(\epsilon_0)>0$,
then $\psi_0=0$. In either case, we get $\psi_0=0 \ \forall \bx \in \mathbb{R}^3/\overline{D}$, and $\psi=0 \ \forall \bx \in D$.

The proof for the static case $\omega = 0$ is standard \cite{CK1}.
\end{proof}

\begin{lemma}\label{comp1}
Let $\bE,\bE_0$ be a solution of the homogeneous vector electric transmission problem.
Then $\psi=\nabla\cdot\bE,\psi_0=\nabla\cdot\bE_0$ satisfy the homogeneous scalar 
electric transmission problem.
\end{lemma}

\begin{proof}
Let $\bE_0,\bE$ be a solution of the homogeneous vector transmission problem.
Clearly, the functions $\psi_0=\nabla\cdot\bE_0,\psi=\nabla\cdot\bE$ satisfy 
the scalar Helmholtz equations $\Delta\psi_0+k_0^2\psi_0=0,\Delta\psi+k^2\psi=0$ in 
the corresponding regions, as well as the regularity conditions 
$\psi_0\in C^2(\mathbb{R}^3/\overline{D})\cap C(\mathbb{R}^3/D)$, 
$\psi\in C^2(D)\cap C(\overline{D})$. From \eqref{cond_sE1} we have
\begin{equation}\label{cond_s1}
\psi_0-\psi=0.
\end{equation}

Let us now assume that 
$\bn\times\nabla\times\bE_0$ and $\bn\times\nabla\times\bE$ are sufficiently smooth and 
that we can take the surface divergence. Then, from 
\eqref{cond_vE2}, \eqref{cond_sE2} and a little algebra, it is easy to check that
\begin{equation}
\frac{1}{\mu_0}\frac{\partial\psi_0}{\partial n}-
\frac{1}{\mu}\frac{\partial\psi}{\partial n}=0.
\end{equation}
If $\bn\times\nabla\times\bE_0$ and $\bn\times\nabla\times\bE$ are merely continuous,
let us define the following quantity:
\begin{equation}
I_1=\lim_{h\rightarrow 0^+}\int_{\partial D}\bn\cdot\Big(\frac{\nabla\times\nabla\times\bE_0}{\mu_0}(\bx+h\bn)-\frac{\nabla\times\nabla\times\bE}{\mu_0}(\bx-h\bn)\Big)w(\bx)ds \, ,
\end{equation}
where 
$\lim_{h\rightarrow 0^+}$ implies taking the limit with $h>0$.
Using the interface condition \eqref{cond_sE2} it is easy to show that
\begin{equation}
\begin{aligned}
I_1&=\lim_{h\rightarrow 0^+}\int_{\partial D}\bn\cdot\Big(\frac{\nabla\nabla\cdot\bE_0}{\mu_0}(\bx+h\bn)-\frac{\nabla\nabla\cdot\bE}{\mu_0}(\bx-h\bn)\Big)w(\bx)ds=\\
&=\lim_{h\rightarrow 0^+}\int_{\partial D}\Big(\frac{1}{\mu_0}\frac{\partial \psi_0}{\partial n}(\bx+h\bn)-\frac{1}{\mu}\frac{\partial\psi}{\partial n}(\bx-h\bn)\Big)w(\bx)ds.
\end{aligned}
\end{equation}
On the other hand, integration by parts yields
\begin{equation}
\begin{aligned}
I_1&=\lim_{h\rightarrow 0^+}\int_{\partial D}\bn\cdot\Big(\frac{\nabla\times\nabla\times\bE_0}{\mu_0}(\bx+h\bn)-\frac{\nabla\times\nabla\times\bE}{\mu_0}(\bx-h\bn)\Big)w(\bx)ds=\\
&=\lim_{h\rightarrow 0^+}\int_{\partial D}-\nabla_s\cdot\Big(\frac{\bn\times\nabla\times\bE_0}{\mu_0}(\bx+h\bn)-\frac{\bn\times\nabla\times\bE}{\mu_0}(\bx-h\bn)\Big)w(\bx)ds=\\
&=\lim_{h\rightarrow 0^+}\int_{\partial D}\Big(\frac{\bn\times\nabla\times\bE_0}{\mu_0}(\bx+h\bn)-\frac{\bn\times\nabla\times\bE}{\mu_0}(\bx-h\bn)\Big)\cdot\nabla w(\bx)ds=0,
\end{aligned}
\end{equation}
where $\nabla_s\cdot \J$ is the surface divergence of the tangential vector field
$\J$, and we have used
\eqref{cond_vE2} in the last step.
Note that the reason we require this weak formulation is that, in general, 
we cannot interchange the limit $h\rightarrow 0$ with the surface divergence.
\end{proof}

\begin{lemma}\label{comp2}
Let $\bH,\bH_0$ be a solution of the homogeneous vector magnetic transmission problem.
Then $\varphi=\nabla\cdot\bH,\varphi_0=\nabla\cdot\bH_0$ satisfy the homogeneous 
scalar magnetic transmission problem.
\end{lemma}

\begin{proof}
The proof is analogous to that of Lemma \ref{comp1}.
\end{proof}

We are now in a position to prove the desired uniqueness result, namely Theorem
\ref{uniquethm}. 

\section{Proof of Theorem \ref{uniquethm}} \label{pdeuniq}

\noindent
{\rm Theorem}:
{\em The vector electric and magnetic transmission problems have unique solutions
for $\omega \ge0$.}

\begin{proof}
Let us assume that $\omega>0$ and that $\bE_0,\bE$ are a solution of the homogeneous 
vector electric transmission problem. 
Then the functions $\psi_0=\nabla\cdot\bE_0,\psi=\nabla\cdot\bE$ satisfy 
the homogeneous scalar electric transmission problem, and $\psi_0=\psi=0$ from
Lemma \ref{scalunique}.

Thus, $\bE_0$ and $\bE$, together with 
$\bH_0:=\frac{\nabla\times\bE_0}{i\omega\mu_0}, \bH:=\frac{\nabla\times\bE}{i\omega\mu}$,
constitute an electromagnetic field that satisfies the usual homogeneous interface
conditions (continuity fn the tangential electric and magnetic fields). 
From the uniqueness theorem \ref{thm_uni_vector}, we may conclude that 
$\bE_0=0, \forall \bx \in \mathbb{R}^3/\overline{D}$ and $\bE=0, \forall \bx\in D$.

For the static case, $\omega=0$, the theorem is a consequence of the fact that the 
operators $D_0, S'_0$ and $M_0$ have a spectrum contained in the 
interval $[-\frac{1}{2},\frac{1}{2}]$ (\cite{CK1}, Chapter 5).
Due to the regularity properties of the solution at the boundary, 
in particular the fact that $\nabla\cdot\bE_0, \nabla \times\bE_0 \in C(\mathbb{R}^3/D)$ 
and $\nabla\cdot\bE,\nabla\times\bE \in C(\overline{D})$, we can apply the 
following representation formula for Laplace vector fields 
(\cite{CK1}, Theorems 4.11, 4.13):
\begin{equation}\label{rep_laplace}
\begin{aligned}
&\nabla\times S_0[\bn\times\bE]-S_0[\bn (\nabla\cdot\bE)]+
S_0[\bn\times\nabla\times\bE]-\nabla S_0[\bn\cdot\bE]=-\bE & \bx\in D\\
&\nabla\times S_{0}[\bn\times\bE_0]-
S_{0}[\bn (\nabla\cdot\bE_0)]+S_{0}[\bn\times\nabla\times\bE_0]-
\nabla S_{0}[\bn\cdot\bE_0]=\bE_0 & \bx\in \mathbb{R}^3/\overline{D} \, .
\end{aligned}
\end{equation}
Taking the curl, we obtain
\begin{equation}
\begin{aligned}
&\nabla\times\nabla\times S_0[\bn\times\bE]-
\nabla\times S_0[\bn (\nabla\cdot\bE)]+\nabla\times 
S_0[\bn\times\nabla\times\bE]=-\nabla\times\bE & \bx\in D\\
&\nabla\times\nabla\times S_0[\bn\times\bE_0]-
\nabla\times S_0[\bn (\nabla\cdot\bE_0)]+\nabla\times 
S_0[\bn\times\nabla\times\bE_0]=\nabla\times\bE_0 & \bx\in \mathbb{R}^3/\overline{D} \, .
\end{aligned}
\end{equation}
Substracting the tangential components of 
$\nabla\times\bE$ from those of $\nabla\times\bE_0$,
and using the jump conditions 
\eqref{cond_vE1}, \eqref{cond_vE2} and \eqref{cond_sE1},
we obtain
\begin{equation}
\begin{aligned}
&-\frac{\mathbf{f}}{2}\Big(1+\frac{\mu}{\mu_0}\Big)+M_0\Big[\mathbf{f}\Big(1-\frac{\mu}{\mu_0}\Big)\Big]=0,
\end{aligned}
\end{equation}
where $\mathbf{f}=\bn\times\nabla\times\bE_0$. Thus,
\begin{equation}
\begin{aligned}
&-\frac{\mathbf{f}}{2}\Big(\frac{\mu_0+\mu}{\mu_0-\mu}\Big)+M_0(\mathbf{f})=0.
\end{aligned}
\end{equation}
If $\mathbf{f}\ne 0$, then $\Big(\frac{\mu_0+\mu}{\mu_0-\mu}\Big)$ must be a real number,
say $\lambda$, between $-1$ and $1$ and
\begin{equation}
\begin{aligned}
\Big(\frac{\mu_0+\mu}{\mu_0-\mu}\Big)=\lambda\Rightarrow\mu=\mu_0\Big(\frac{\lambda-1}{\lambda+1}\Big) .
\end{aligned}
\end{equation}
As $\mu_0>0$, we must have $\mu\le 0$, a contradiction. Thus,
$\mathbf{f}=\bn\times\nabla\times\bE_0=0$. By the jump condition \eqref{cond_vE2}
and the fact that boundary data is homogeneous, 
we have $\bn\times\nabla\times\bE=0$.
The representation formulas \eqref{rep_laplace} are now simplified, 
taking the form
\begin{equation}
\begin{aligned}
&\nabla\times S_0[\bn\times\bE]-S_0[\bn (\nabla\cdot\bE)]-\nabla S_0[\bn\cdot\bE]=-\bE & \bx\in D\\
&\nabla\times S_{0}[\bn\times\bE_0]-S_{0}[\bn(\nabla\cdot\bE_0)]-\nabla S_{0}[\bn\cdot\bE_0]=\bE_0 & \bx\in \mathbb{R}^3/\overline{D} \, .
\end{aligned}
\end{equation}
Taking the divergence, we have
\begin{equation}
\begin{aligned}
&D_0[\nabla\cdot\bE]=-\nabla\cdot\bE & \bx\in D\\
&D_{0}[\nabla\cdot\bE_0]=\nabla\cdot\bE_0 & \bx\in \mathbb{R}^3/\overline{D} \, .
\end{aligned}
\end{equation}
Substracting the limiting values from the corresponding sides we obtain
\begin{equation}
\begin{aligned}
-\frac{\nabla\cdot\bE_0+\nabla\cdot\bE}{2}+D_0[\nabla\cdot\bE_0-\nabla\cdot\bE]=0&\ \ \ \bx\in \partial D.
\end{aligned}
\end{equation}
Using the jump condition \eqref{cond_sE1}, we have that 
$\nabla\cdot\bE_0=\nabla\cdot\bE=0 \ \forall \bx\in \partial D$.
Thus, the representation formulas \eqref{rep_laplace} simplify even further to
\begin{equation}
\begin{aligned}
&\nabla\times S_0[\bn\times\bE]-\nabla S_0[\bn\cdot\bE]=-\bE & \bx\in D\\
&\nabla\times S_{0}[\bn\times\bE_0]-\nabla S_{0}[\bn\cdot\bE_0]=\bE_0 & \bx\in \mathbb{R}^3/\overline{D} \, .
\end{aligned}
\end{equation}
Multiplying the first equation by $\epsilon $ and the second by $\epsilon_0$, substracting the tangential components and using the jump condition \eqref{cond_sE2}, we obtain
\begin{equation}
-\mathbf{g}\Big(\frac{\epsilon_0+\epsilon}{2}\Big)+M_0[\mathbf{g}(\epsilon_0-\epsilon)]=0,
\end{equation}
where $\mathbf{g}=\bn\times\bE_0$. By the same contradiction argument as above, 
we find that $\mathbf{g}=0$, so that $\bn\times\bE_0=0,\bn\times\bE=0$.
Thus, the representation formula \eqref{rep_laplace} is simplified again:
\begin{equation}
\begin{aligned}
&-\nabla S_0[\bn\cdot\bE]=-\bE & \bx\in D\\
&-\nabla S_{0}[\bn\cdot\bE_0]=\bE_0 & \bx\in \mathbb{R}^3/\overline{D}.
\end{aligned}
\end{equation}
Substracting the normal components, we have
\begin{equation}
\begin{aligned}
&-\frac{p}{2}\Big(1+\frac{\epsilon_0}{\epsilon}\Big)-S'_0\Big[p\Big(1-\frac{\epsilon_0}{\epsilon}\Big)\Big]=0,
\end{aligned}
\end{equation}
where $p=\bn\cdot\bE_0$. Thus,
\begin{equation}
\begin{aligned}
&-\frac{p}{2}\Big(\frac{\epsilon_0+\epsilon}{\epsilon_0-\epsilon}\Big)+S'_0[p]=0.
\end{aligned}
\end{equation}
From the spectral properties of $S'$, a simple contradiction argument shows that
$\bn\cdot\bE_0=\bn\cdot\bE=0$. Finally, from the representation formula
\eqref{rep_laplace}, we have
$\bE=0\   \forall\bx\in D$ and $\bE_0=0\ \forall\bx\in \mathbb{R}^3/\overline{D}$.

The proof for the magnetic vector transmission problem is essentially the same.
\end{proof}

\section{Existence of solutions to the vector transmission problems}
\label{pdeexist}

\begin{thm} \label{evecexist}
The vector electric transmission problem has a solution for any $\omega \ge0$.
\end{thm}

\begin{proof}
Let us consider the representation for the solution $\bE_0,\bE$ given by
\eqref{eerepresentation} and the corresponding integral equation \eqref{veteq},
obtained by imposing the 
interface conditions \eqref{cond_vE1} - \eqref{cond_sE2}.
We examine this Fredholm equation on the 
product space $C_t^{0,\alpha}(Div,\partial D)\times C^{0,\alpha}(\partial D)
\times C_t^{0,\alpha}(\partial D)\times C^{0,\alpha}(\partial D)$,
using the norm 
\[ \|(\mathbf{a},\sigma,\mathbf{b},\rho)\| \equiv
\max\big(\|\mathbf{a}\|_{\alpha,\partial D}+
\|\nabla_s\cdot\mathbf{a}\|_{\alpha,\partial D},\|\sigma\|_{\alpha,\partial D},
\|\mathbf{b}\|_{\alpha,\partial D},\|\rho\|_{\alpha,\partial D}\big) \, ,
\]
following the approach of \cite{CK2}.
It will be convenient to rewrite 
\eqref{veteq} in the form
\begin{equation}
\mathbf{(B+K)x=y}
\label{veteqshort}
\end{equation}
where
\begin{equation}\label{Op_B}
\mathbf{B} \equiv \left( \begin{array}{cccc}
I_1 & B_{12} & B_{13} & 0 \\
0 & I_2 & B_{23} & 0 \\
0 & 0 & I_3 & 0 \\
B_{41} & 0 & 0 & I_4 
\end{array} \right)  \, ,\quad
\mathbf{K} \equiv \left( \begin{array}{cccc}
K_{11} & 0 & 0 & K_{14} \\
0 & K_{22} & 0 & K_{24} \\
K_{31} & K_{32} & K_{33} & 0 \\
0 & K_{42} & K_{43} & K_{44}
\end{array} \right),
\end{equation}
\begin{equation}
\mathbf{x} \equiv \left( \begin{array}{c}
\ba \\  \sigma \\ \bb \\ \rho
\end{array} \right), \ \ 
\mathbf{y} \equiv
\left( \begin{array}{c}
\mathbf{f} \\  q \\ \mathbf{g} \\ p
\end{array} \right) \, .
\end{equation}
Here,
\begin{equation}
\begin{array}{ll}
B_{12}=-\bn\times\big(\mu_0 S_{k_0}-\mu S_{k}\big)(\bn\sigma), &
B_{13}=\bn\times\big( \mu_0\epsilon_0S_{k_0}-\mu\epsilon S_{k}\big)(\bb)\\
B_{23}=\nabla\cdot\Big(\mu_0\epsilon_0S_{k_0}-\mu\epsilon S_{k}\Big)(\bb), &
B_{41}=\bn\cdot\nabla\times\big(\epsilon_0\mu_0S_{k_0}-\epsilon\mu S_{k}\big)(\ba)\\
K_{11}=\big(\mu_0M_{k_0}-\mu M_{k}\big)(\ba),  &
K_{14}=\bn\times\nabla\big(S_{k_0}-S_k\big)(\rho)\\
K_{22}=\Big(\mu_0D_{k_0}-\mu D_{k}\Big)(\sigma), & 
K_{24}=-\omega^2\Big(\mu_0\epsilon_0S_{k_0}-\mu\epsilon S_{k}\Big)(\rho)\\
K_{31}=\bn\times\nabla\times\nabla\times\big(S_{k_0}-S_k\big)(\ba),& 
K_{32}=-\bn\times\nabla\times\big(S_{k_0}-S_k\big)(\bn\sigma)\\
K_{33}=\big(\epsilon_0 M_{k_0}-\epsilon M_{k}\big)(\bb),& 
K_{42}=-\bn\cdot\big(\epsilon_0\mu_0S_{k_0}-\epsilon\mu S_{k}\big)(\bn\sigma)\\
K_{43}=\bn\cdot\big(\mu_0\epsilon_0^2S_{k_0}-\mu\epsilon^2 S_{k}\big)(\bb),& 
K_{44}=\big(\epsilon_0S'_{k_0}-\epsilon S'_{k}\big)(\rho) 
\end{array}
\end{equation}
with
\[
I_1(\ba) =\frac{\mu_0+\mu}{2}\ba,\ 
I_2(\sigma) =\frac{\mu_0+\mu}{2}\sigma,\ 
I_3(\bb) =\frac{\epsilon_0+\epsilon}{2}\bb,\  
I_4(\rho) =-\frac{\epsilon_0+\epsilon}{2}\rho \, .
\]
Note that the operators $I_1,I_2,I_3,I_4$ are diagonal and invertible 
from \eqref{microwave_cond}. We will denote their inverses by 
${\cal I}_1,{\cal I}_2,{\cal I}_3,{\cal I}_4$.
Note also that the operator $\mathbf{K}$ is compact and that $\mathbf{B}$ is invertible 
with
\begin{equation}
\mathbf{B}^{-1}=
\left( \begin{array}{cccc}
{\cal I}_1 & -{\cal I}_1 {\cal I}_2B_{12} & {\cal I}_1 {\cal I}_2 {\cal I}_3(B_{12}B_{23}-B_{13}I_2) & 0 \\
0 & {\cal I}_2 & -{\cal I}_2{\cal I}_3B_{23} & 0 \\
0 & 0 & {\cal I}_3 & 0 \\
-{\cal I}_1{\cal I}_4B_{41} & {\cal I}_1{\cal I}_2{\cal I}_4B_{12}B_{14} & 
{\cal I}_1{\cal I}_2{\cal I}_3{\cal I}_4(B_{12}B_{23}B_{41}-B_{13}B_{41}I_2) & {\cal I}_4 
\end{array} \right).
\end{equation}
Thus, the Fredholm alternative can be applied and it remains only to show that,
in \eqref{veteqshort}, $\mathbf{y}=0$ implies $\mathbf{x}=0$.

For this, we begin by observing that
the corresponding fields given by \eqref{eerepresentation} satisfy the 
homogeneous vector electric transmission problem. 
By Theorem \eqref{uniquethm}, they must be $\bE_0=\bE=0$.
We now define the following auxiliary fields:
\begin{equation}
\begin{aligned}
\widetilde{\bE_0}&=-\frac{1}{\mu_0}\Big(\mu_0\nabla\times S_{k_0}(\ba)-\mu_0S_{k_0}(\bn \sigma )+\mu_0\epsilon_0 S_{k_0}(\bb)+\nabla S_{k_0}(\rho)\Big) & \bx\in D\\
\widetilde{\bE}&=\frac{1}{\mu}\Big(\mu\nabla\times S_{k}(\ba)-\mu S_{k}(\bn \sigma )+\mu\epsilon S_{k}(\bb)+\nabla S_{k}(\rho)\Big)& \bx\in \mathbb{R}^3/\overline{D}.
\end{aligned}
\end{equation}
Note that $\widetilde{\bE_0}$ is defined in the interior, 
using the exterior material parameters, and that 
$\widetilde{\bE}$ is defined in the exterior, using the interior material
parameters.
From the jump conditions for layer potentials 
\eqref{Skdef1}, \eqref{Skdef2}, \eqref{Dkdef}, \eqref{Mkdef}, we have
\begin{equation}
\begin{aligned}
\bn\times\big(\frac{\bE_0}{\mu_0}+\widetilde{\bE_0}\big)&=\ba\\
\bn\times\big(\frac{\nabla\times\bE_0}{\mu_0}+\nabla\times\widetilde{\bE_0}\big)&= \epsilon_0\bb\\
\bn\cdot \big(\frac{\bE_0}{\mu_0}+\widetilde{\bE_0}\big)&=-\frac{\rho}{\mu_0}\\
\frac{\nabla\cdot\bE_0}{\mu_0}+\nabla\cdot\widetilde{\bE_0}&=\sigma,
\end{aligned}
\end{equation}
and
\begin{equation}
\begin{aligned}
\bn\times\big(\widetilde{\bE}-\frac{\bE}{\mu}\big)&=\ba\\
\bn\times\big(\nabla\times\widetilde{\bE}-\frac{\nabla\times\bE}{\mu}\big)&=\epsilon\bb\\
\bn\cdot \big(\widetilde{\bE}-\frac{\bE}{\mu}\big)&=-\frac{\rho}{\mu}\\
\nabla\cdot\widetilde{\bE}-\frac{\nabla\cdot\bE}{\mu}&=\sigma.
\end{aligned}
\end{equation}
These, together with the fact that $\bE_0=\bE=0$, shows that
the fields $\widetilde{\bE_0},\widetilde{\bE}$ satisfy the jump conditions
\begin{equation}
\bn\times\Big(\widetilde{\bE_0}-\widetilde{\bE}\Big)=0\\
\end{equation}
\begin{equation}
\bn\times\Big(\frac{\nabla\times\widetilde{\bE_0}}{\epsilon_0}-\frac{\nabla\times\widetilde{\bE}}{\epsilon}\Big)=0
\end{equation}
\begin{equation}
\bn\cdot\big(\mu_0 \widetilde{\bE_0}-\mu\widetilde{\bE}\big)=0
\end{equation}
\begin{equation}
\nabla\cdot\widetilde{\bE_0}-\nabla\cdot\widetilde{\bE}=0.
\end{equation}
By Theorem \ref{uniquethm}, for the vector magnetic transmission problem 
(with material properties swapped), we have that 
$\widetilde{\bE_0}=\widetilde{\bE}=0$.
Finally, from the jump conditions above, we may conclude that 
$\ba=\bb=0$, $\sigma=\rho=0$.
\end{proof}

\begin{thm}
The vector magnetic transmission problem has a solution for any $\omega \ge0$.
\end{thm}

\begin{proof}
The proof is analogous to that of 
Theorem \ref{evecexist}.
\end{proof}

\section{Proof of Theorem \ref{stabthm}}
\label{pdestab}

\begin{proof}
It is easy to see that the 
operators $\mathbf{B}$ and $\mathbf{K}$  in \ref{Op_B} 
are bounded and compact in the function space 
$C_t^{0,\alpha}(\partial D)\times C^{0,\alpha}(\partial D)\times 
C_t^{0,\alpha}(\partial D)\times C^{0,\alpha}(\partial D)$ 
with the norm $\|(\mathbf{a},\sigma,\mathbf{b},\rho)\|:=
\max\big(\|\mathbf{a}\|_{\alpha,\partial D},
\|\sigma\|_{\alpha,\partial D},
\|\mathbf{b}\|_{\alpha,\partial D},
\|\rho\|_{\alpha,\partial D}\big)$. 
As $C_t^{0,\alpha}(Div,\partial D)$ is dense in $C_t^{0,\alpha}(\partial D)$, 
the operator $\mathbf{B+K}$ has the same nullspace as in the space analyzed
in the proof of Theorem \ref{evecexist} in Appendix \ref{pdeexist}.
Thus, it is uniquely solvable. Its inverse operator exists and is 
bounded as a map from 
$C_t^{0,\alpha}(\partial D)\times C^{0,\alpha}(\partial D)\times 
C_t^{0,\alpha}(\partial D)\times C^{0,\alpha}(\partial D)$ to  
$C_t^{0,\alpha}(\partial D)\times C^{0,\alpha}(\partial D)\times 
C_t^{0,\alpha}(\partial D)\times C^{0,\alpha}(\partial D)$. 
Due to the collective compactness of the operators, the 
continuity of the inverse is uniform in $[0,\omega_{\max}]$. 
This, together with the regularity properties of the operators in 
\eqref{eerepresentation} \cite{CK1} implies the desired estimates 
\eqref{ineq_1}, \eqref{ineq_2}, \eqref{ineq_3}, \eqref{ineq_4}.
\end{proof}

\bibliographystyle{plain}

\end{document}